\documentclass[manuscript,screen]{acmart}
\usepackage{tabularx}
\usepackage{booktabs}
\usepackage{multirow}
\usepackage{graphicx}
\usepackage{fancyhdr}
\usepackage{color}
\usepackage{colortbl}
\usepackage[normalem]{ulem}
\useunder{\uline}{\ul}{}
\newtheorem{problem}{Problem}
\newtheorem{assumption}{Assumption}
\newtheorem{lamma}{\textbf{Lamma}}
\newtheorem{proposition}{\textbf{Proposition}}
\newcommand{\dataset}[1]{\(\mathtt{#1}\)}
\AtBeginDocument{%
  }

\setcopyright{acmlicensed}
\copyrightyear{2018}
\acmYear{2018}
\acmDOI{XXXXXXX.XXXXXXX}

\acmConference[Conference acronym 'XX]{Make sure to enter the correct
  conference title from your rights confirmation emai}{June 03--05,
  2018}{Woodstock, NY}
\acmISBN{978-1-4503-XXXX-X/18/06}

\begin{document}

\title{Causal Structure Representation Learning of Unobserved Confounders in Latent Space for Recommendation}

\author{Hangtong Xu}
\affiliation{%
  \institution{MIC Lab, College of Computer Science and Technology, Jilin University}
  \city{Changchun}
  \country{China}
  }
\email{xuht21@mails.jlu.edu.cn}
\author{Yuanbo Xu}
\affiliation{%
  \institution{MIC Lab, College of Computer Science and Technology, Jilin University}
  \city{Changchun}
  \country{China}
  }
\email{yuanbox@jlu.edu.cn}
\authornote{Corresponding Author: Yuanbo Xu}
\author{Chaozhuo Li}
\affiliation{%
  \institution{ Key Laboratory of Trustworthy Distributed Computing and Service (MoE), Beijing University of Posts and Telecommunications}
  \country{China}
  }
\email{lichaozhuo@bupt.edu.cn}

\author{Fuzhen Zhuang}
\affiliation{%
  \institution{Institute of Artificial Intelligence, Beihang University and Zhongguancun Laboratory}
  \city{Beijing}
  \country{China}
  }
\email{zhuangfuzhen@buaa.edu.cn}

\renewcommand{\shortauthors}{Hangtong et al.}

\begin{abstract}
 Inferring user preferences from users' historical feedback is a valuable problem in recommender systems. Conventional approaches often rely on the assumption that user preferences in the feedback data are equivalent to the real user preferences without additional noise, which simplifies the problem modeling. However, there are various confounders during user-item interactions, such as weather and even the recommendation system itself. Therefore, neglecting the influence of confounders will result in inaccurate user preferences and suboptimal performance of the model. Furthermore, the unobservability of confounders poses a challenge in further addressing the problem. Along these lines, we refine the problem and propose a more rational solution to mitigate the influence of unobserved confounders. Specifically, we consider the influence of unobserved confounders, disentangle them from user preferences in the latent space, and employ causal graphs to model their interdependencies without specific labels. By ingeniously combining local and global causal graphs, we capture the user-specific effects of confounders on user preferences. Finally, we propose our model based on Variational Autoencoders, named \textbf{C}ausal \textbf{S}tructure \textbf{A}ware \textbf{V}ariational \textbf{A}uto\textbf{e}ncoders (CSA-VAE) and theoretically demonstrate the identifiability of the obtained causal graph. We conducted extensive experiments on one synthetic dataset and nine real-world datasets with different scales, including three unbiased datasets and six normal datasets, where the average performance boost against several state-of-the-art baselines achieves up to 9.55\%, demonstrating the superiority of our model. Furthermore, users can control their recommendation list by manipulating the learned causal representations of confounders, generating potentially more diverse recommendation results. Our code is available at \href{https://github.com/MICLab-Rec/CSA}{Code-link}\footnote{https://github.com/MICLab-Rec/CSA}.
\end{abstract}



\received{20 February 2007}
\received[revised]{12 March 2009}
\received[accepted]{5 June 2009}

\maketitle

\section{Introduction}
\begin{figure}[!tb]
    \centering
    \includegraphics[width = \linewidth]{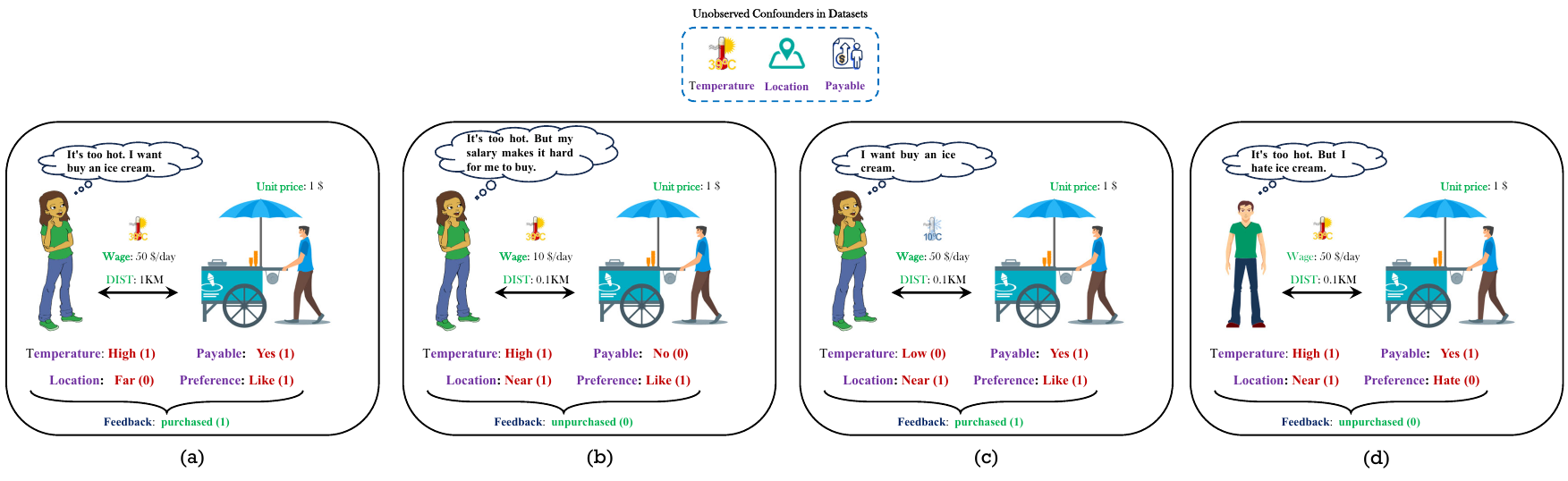}
    \caption{An example illustrating that user preferences in the feedback data are influenced by both the users themselves and external confounders(e.g., temperature, location).}
    \label{fig:1}
\end{figure}
Recommender systems play a vital role in information technology, aiming to assist users in discovering content or products that might align with their interests \cite{xv3,xv2}. User historical feedback data is a crucial basis for model prediction of user preferences. Based on this, many highly effective methods have been proposed, such as MultiVAE \cite{mutivae}, MacridVAE \cite{mrcrid}, etc. Existing work often assumes that user preferences in historical feedback data are noise-free, as shown in Figure \ref{fig:1} (a) and (d), the alignment of user preferences (like or hate) with the collected feedback data (purchased or unpurchased) is crucial, thus focusing on fitting the feedback data can achieve acceptable model performance. \par

However, various confounders inevitably influence users during the interaction process, affecting their final decisions, and the extent of the impact differs among users \cite{survey1,survey2}. As shown in Figure \ref{fig:1}, user preferences are influenced by both the intrinsic characteristics of users (\textsf{user-specific preference}) and external confounders in the interaction environment (e.g., temperature, location). For example, when analyzing user feedback for a specific product like ice cream on an online retail platform, we observe varying user preferences due to confounders. Some users may exhibit opposite feedback regarding their preferences for ice cream. For instance, the wage at the time of purchase is a confounder. As shown in Figure \ref{fig:1} (b), it's a tough decision for users to spend a tenth of their daily salary on ice cream. On the other hand, when user preference is strong enough, confounders have less impact on the final decision. As Figure \ref{fig:1} (c) shows, the user gives positive feedback for ice cream due to their favorite flavor on cold days. \par
Nevertheless, relying solely on confounders to determine user preferences is also unreasonable \cite{xu1}. Figures \ref{fig:1} (a) and (d) show that two users with the same confounder conditions give opposite feedback. Thus, it is evident that user preferences, as reflected in the feedback data, are a combination of the intrinsic characteristics of users and external confounders like wage, representing the interplay of intrinsic and extrinsic influences. This also explains why models improve performance when additional information, such as time and location, is incorporated. Unfortunately, most confounders are unobservable, and we cannot obtain corresponding labels from users as additional information. Hence, we cannot explicitly model confounders to separate them from user preferences, and addressing the dynamic impact of confounders on users remains a challenge.\par

To address these challenges, we first reformulate the user preference prediction problem by introducing the influence of confounders. In this way, user preferences in the feedback data stem from the combined impact of the inherent preferences of users and external confounders. We proposed a mild assumption of confounder independence to disentangle confounders from user preferences in the latent space. Specifically, we assumed that the influence of all confounders on each user originates from the same set of confounders, which ensures user independence of confounders. Furthermore, we found that confounders are not independent of each other. For example, the weather depends on users' location. We utilized a causal structural model (SEM) to represent the generation process of confounders. Specifically, we employed a binary matrix to denote the global causal graph, where directed edges signify the dependency relationships between confounders. However, for different users, the relationships between confounders also vary, sometimes opposite. Thus, we used an additional local causal graph to capture the user-specificity of confounders.

Furthermore, we demonstrate that the learned causal representations of confounders are controllable, potentially offering users fine-grained control over the objectives of their recommendation lists with the learned causal graphs. Finally, we combine the obtained causal representations of confounders with the inherent preferences of users to fit user preferences in historical feedback data, proposed a model based on Variational Autoencoder (VAE) named \textbf{C}ausal \textbf{S}tructure \textbf{A}ware \textbf{V}ariational \textbf{A}uto\textbf{e}ncoders (CSA-VAE) to learn causal representations of user preferences and confounders simultaneously. In addition, we theoretically proved the identifiability of the model.\par

The contributions of our work can be summarized as follows: 
\begin{itemize}
    \item We have re-formalized the problem of user preference prediction, providing a more reasonable modeling approach for user preferences.
    \item We introduced a mild assumption that allows for the independent representation of user preferences and the influence of confounders and utilized causal graphs to capture the dependencies among confounders. Furthermore, we provide proof of the identifiability of the causal graph and the visualization of learned confounder representations.
    \item We employed global and local causal graphs to capture the invariance and specificity between users and confounders. We proposed a model based on the Variational Autoencoder (VAE) to simultaneously learn causal representations of user preferences and confounders.
    \item  We conducted extensive experiments on a synthetic dataset and nine real-world datasets with different scales, including three unbiased datasets and six normal datasets, where the average performance boost against several state-of-the-art baselines achieves up to 9.55\%, demonstrating our model's effectiveness.
    \item \textcolor{black}{We formalize the user-controlled recommendation task by integrating both latent preferences and confounders while incorporating causal interventions to give the user control over their preferences within the recommendation system.}
\end{itemize}
Our code is publicly available at: \href{https://github.com/MICLab-Rec/CSA}{https://github.com/MICLab-Rec/CSA}.

\section{Related Work}

\subsection{Deconfound in Recommendation} 
With the increasing popularity of causal inference as a method to mitigate bias in recommender systems \cite{xv1}, researchers are paying more attention to the challenges posed by confounding biases. Confounding bias is prevalent in recommender systems due to various confounders. While some studies have addressed specific confounding biases, such as item popularity \cite{ref09,ref10,xu2}, many unobservable confounders may also exist. The mainstream approaches can be broadly categorized into two types: (1) \cite{ref05,ref06} utilize additional signals as instrumental or proxy variables to mitigate confounding bias. (2) \cite{ref07,ref08} consider a multiple-treatment setting and infer surrogate confounders from user exposure, incorporating them into the preference prediction model. (3)  SEM-MacridVAE \cite{sem1}, CaD-VAE \cite{sem2} and PlanRec \cite{sem3} utilize additional signals to learn the latent causal structure of confounders and make recommendation.\par
However, they did not address the challenge of confounders in the absence of corresponding labels. CSA-VAE simultaneously learned user causal preferences and graphs without corresponding labels and employed global and local causal graphs to capture the invariance and specificity between users and confounders.

\subsection{Causal Structure Learning}  
We refer to causal representations constructed by causal graphs as causal representations. Over the past few decades, discovering causal graphs from purely observational data has garnered significant attention. \cite{ref11} proposed NOTEARs with a fully differentiable DAG constraint for causal structure learning, \cite{ref12}show the identifiability of learned causal structure from interventional data. The community has raised interest in combining causality and disentangled representation, and \cite{ref13} proposed a method called CausalGAN, which supports ”do-operation” on images, but it requires the causal graph given as a prior.\par
We draw on key ideas from causal structure learning to enhance the application of latent structure learning in recommendations. Additionally, we identify a key challenge in applying shared latent structure to recommendations: confounding factors affect users differently. To address this, we design a personalized structure learning framework and personalized recommendations.

\section{Methodology}
This section thoroughly introduces our proposed model with the necessary theoretical proofs.\par

\subsection{A More Rational Architecture}
Predicting user preferences from their historical interaction data is a common training paradigm in the domain of recommender systems and is generally based on the fundamental assumption that the user preferences contained in the feedback data reflect the true preferences of the user. Based on such assumptions, the user preference prediction problem can be formulated as follows: \par
\begin{figure}[!tb]
    \centering
    \includegraphics[width = 0.6\linewidth]{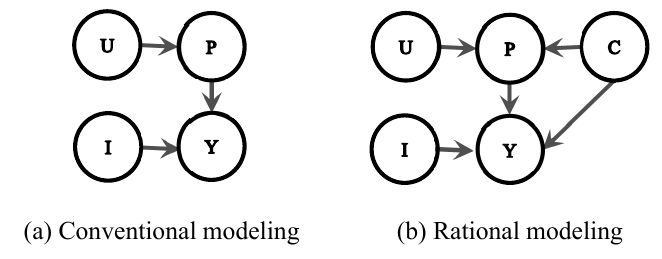}
    \caption{Conventional modeling versus a more rational modeling approach. Left is the conventional modeling, and right is the rational modeling proposed in this work. \textbf{U}$\rightarrow$User; \textbf{I}$\rightarrow$Item; \textbf{C}$\rightarrow$Confounders; \textbf{P}$\rightarrow$Preference in feedback data; \textbf{Y}$\rightarrow$User feedback. The exogenous variables of nodes (e.g., \textbf{C}) are not displayed in the graph.}
    \label{fig:2}
\end{figure}
\begin{problem}
  Given the historical interaction data\footnote{$x_{u, i} \in \{0, 1\}, x_{u, i} = 1$ representing an interaction between user $u$ and item $i$, n is the item num.} $x_{u} = \{x_{u, 1}, ..., x_{u, n}\}$, we can predict the user preference $z_{u}$ with the model parameter $\phi$.
  \begin{equation}
      \begin{aligned}
          q_{\phi}(z_{u} | x_{u}). \nonumber \\
      \end{aligned}
  \end{equation}
\end{problem}
As Figure \ref{fig:2} (a) shows, The fundamental assumption of the problem is that the users themselves (\textbf{U}) solely influence the user preferences in the feedback data (\textbf{P}). Numerous outstanding approaches have arisen to address the above problem, and their solutions can be uniformly summarised into one class of solutions, i.e., \ for a user, they assume the observed data is generated from the following distributions:
\begin{equation}
    \begin{aligned}
        p_{\theta}(x_{u}) = E \left[ \int p_{\theta}(x_{u} | z_{u})p_{\theta}(z_{u})dz_{u}\right].
    \end{aligned}
\end{equation}
Such approaches will restore the generation process of the data with outstanding performance, but they overlook the influence of confounders on feedback in user interactions. As a result, they cannot explain why the user feedback on ice cream shows opposite results due to wage, where wage acts as a confounder. Based on the above findings, we naturally consider incorporating confounders into the data generation process as a more reasonable modeling approach. Firstly, we redefine the problem as follows: \par
\begin{problem}
  Given the historical interaction data $x_{u} = \{x_{u, 1}, ..., x_{u, n}\}$, we can predict the user preference $z_{u}, c_{u}$ with the model parameter $\phi$.
  \begin{equation}
      \begin{aligned}
          q_{\phi}(z_{u}, c_{u} | x_{u}). \nonumber \\
      \end{aligned}
  \end{equation}
\end{problem}
A graphical representation is depicted in Figure \ref{fig:2} (b), illustrating that the inner preferences of users and other confounders influence the observed user preferences in feedback data. We must emphasize that we are exclusively considering unobserved confounders in this problem. For observable confounders, we prefer incorporating them as additional input to enhance the performance of models, such as POI information in POI recommendations. In a similar vein, we propose a feasible solution to the aforementioned problem as follows: \par
\begin{equation}
    \label{eq:2}
    \begin{aligned}
        p_{\theta}(x_u) = E_{p_{\theta}(c_{u})} \left[ \int p_{\theta}(x_{u} | z_{u},c_{u})p_{\theta}(z_{u})dz_{u}\right].
    \end{aligned}
\end{equation}
By employing this approach, we can effectively disentangle user preferences from confounders, thus overcoming the limitations of conventional methods. For example, we can use the do-operation to predict user interactions in their current environment, enabling us to answer a counterfactual question such as "Will users prefer ice cream if the weather is warm?". Even when the labels of the relevant confounders are unknown, we can obtain purer representations of users' preferences than composite entities. \par
\subsection{Causal Modeling of Unobserved Confounders}
When users interact with items, they are inevitably influenced by confounders such as weather, location, etc. Hence, many models that utilize additional information, such as location data as supplementary input, often demonstrate better predictive performance. However, observable confounders represent only a small fraction of this vast population. In most cases, confounders are unobservable. Therefore, disentangling the impact of unobservable confounders from feedback data remains challenging. To address this issue, we start by making a mild assumption on the independence of the confounders and give corresponding proofs:
\begin{assumption}
    Given confounders $\mathbf{C} = \{c_{1}, c_{2}, ..., c_{k}\}$, we assume that is independent of the user.
    \begin{equation}
        \begin{aligned}
            c_{j} \perp u_{i}, \hspace{0.2cm} i \in [1, m], j \in [1, k]. \nonumber\\
        \end{aligned}
    \end{equation}
    \label{asu:1}
\end{assumption}
One premise for this assumption is that the data collection environment for user feedback is consistent. This condition is often easily met in recommender systems, where data is sourced from historical interactions within a specific platform over a certain period. Hence, this assumption is reasonably mild and applies to most recommendation scenarios. With this assumption, we can separate the confounders from the inherent preferences of users, thus achieving the model architecture as depicted in Eq \ref{eq:2}.\par
Next, we will provide the corresponding proof of Assumption \ref{asu:1}:
\begin{proof}
    In our paper, the confounder set \textbf{C} includes various unobserved confounders, such as wage. The representation of confounders is not dependent on specific users; thus, the exogenous variables of \textbf{C} are not dependent on the exogenous variables of \textbf{U}. Similarly, the user's inherent preference \textbf{U} will not change due to the point of interest (poi) or other confounders. Thus, the exogenous variables of \textbf{U} are not dependent on the exogenous variables of \textbf{C}. Given the exogenous variables \textbf{$E_U$} $\rightarrow$ \textbf{U} and \textbf{$E_C$} $\rightarrow$ \textbf{C}:
    \begin{equation}
        \begin{aligned}
            \mathbf{\mathrm{E_U}} \perp \mathbf{\mathrm{E_C}}. \nonumber
        \end{aligned}
    \end{equation}
    In the causal graph shown in Figure \ref{fig:2} (b), \textbf{U}$\rightarrow$\textbf{P}$\leftarrow$\textbf{C} is a collider. According to the characteristics of colliders discussed in Section 2.3 of Pearl's book \cite{pearl}, when the condition of independence of the exogenous variables of \textbf{C} and \textbf{U} is met, \textbf{C} and \textbf{U} are independent.
\end{proof}
\textcolor{black}{One of the constraints for the validity of Assumption \ref{asu:1} is that the exogenous variables of the confounders must be independent of the users, which means our model can only account for confounders that do not rely on the user, such as location, weather, etc. However, confounders closely related to the user, such as social relationships and background, cannot be included in our model.}\par.
From this, it can be inferred that the assumption of independence between \textbf{C} and \textbf{U} in the paper is reasonable. In our model, the representations of confounders and user preferences are extracted using different neural networks through feedback data, satisfying the independence assumption mentioned above. Additionally, based on the characteristics of the collider, we can infer that when we condition the preference in the training set, \textbf{C} and \textbf{U} are likely to be dependent. This observation explains the intuition that \textbf{C} and \textbf{U} are correlated and provides theoretical evidence for this correlation. We have noted this characteristic and made special design considerations in our methodology to address this aspect. Thus, we did not directly use the representation of confounders for modeling \textcolor{black}{(randomly initialized embedding representation of confounders in \cite{dag,ref03,sem3})}, but instead employed local and global causal graphs to capture this correlation (details in Section \ref{sec:3}).\par
\begin{figure}[!tb]
    \centering
    \includegraphics[width = 0.6\linewidth]{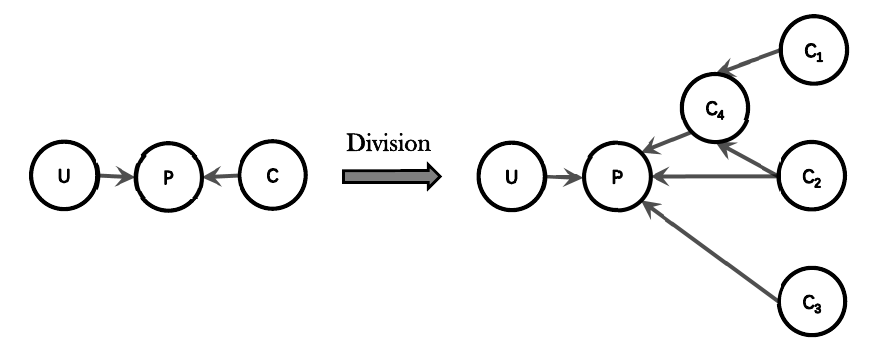}
    \caption{An example showing confounders disentanglement in latent space, where the confounders \textbf{C} are further dismantled to k concepts $\{c_1,c_2,...,c_k\}$ based on the given causal graph.}
    \label{fig:3}
\end{figure}
Moreover, the relationships among confounders are not independent and follow a causal structure represented by a specific causal graph. For instance, consider two confounders: location and item category. The category of items depends on the nature of the location (e.g., $\mathit{location} \rightarrow \mathit{item\ category}$); for example, clothes will be sold in a shopping mall and will not appear in a library. To formalize the causal representation, we consider $k$ confounders in the data. The confounders are causally structured by a Directed Acyclic Graph (DAG) with an adjacency matrix \textbf{A}. For convenience, we adopt a linear Structural Causal Model (SCM) as previous research \cite{dag,ref03} to model the relationship between confounders and the causal graph, as illustrated in Eq \ref{eq:3}:

\begin{equation}
    \label{eq:3}
    \begin{aligned}
        \mathbf{\mathrm{C}} = \mathbf{\mathrm{A}}^{\top}\mathbf{\mathrm{C}} + \epsilon = (I - \mathbf{\mathrm{A}}^{\top})^{-1}\epsilon, \\
    \end{aligned}
\end{equation}
where \textbf{A} is the parameter to be learned by our model and is initialized by the all-one matrix, $\epsilon$ is independent Gaussian noise acting as the exogenous variables of the confounders, and \textbf{C} is a structured causal representation of the $k$ confounders generated by a DAG. This approach can further disentangle the confounders based on the causal graph \textbf{A}, as depicted in Figure \ref{fig:3}. As expected, nonlinear SCM is more suitable for complex scenarios like recommender systems than linear ones. Therefore, in our practical deployment, we utilize the nonlinear SCM, which will be further elucidated in Section \ref{sec:3}.

\subsection{Causal Structure Learning}
\label{sec:3}
As mentioned before, an accurate causal graph enables us to better capture the influence of confounders on user preferences and the dependencies among these confounders. In conventional causal graphs, a causal flow between nodes is typically represented by an adjacency matrix, and the weights in this matrix measure the influence of parent nodes on their respective child nodes. However, when applied in recommendation system scenarios, this approach falls short of capturing the heterogeneity among users.\par
\textcolor{black}{
Consider two different users in the context of music recommendation. One user enjoys listening to different types of music in different locations (e.g., quiet music in the library), while the other prefers the same type of music regardless of location. The impact of location on these users differs: for the former, location influences their music preferences, while the latter remains unaffected by it. Although the location-music causal relationship can be inferred from the global causal graph, using this global graph alone to predict the impact of location on the music preferences of these two users would lead to suboptimal recommendations for the latter. To address this issue, we employ a combination of global and local causal graphs to capture the influence of confounders on user preferences. The global causal graph captures as many causal relationships as possible within the given constraints, and the local causal graph optimizes recommendation performance by masking the irrelevant relations of the global graph that do not impact the current user.}
\subsubsection{\textbf{Global SCM}} Specifically, we use the global causal graph to model the relationships among all confounders. Given the adjacency matrix $\mathcal{G}^{global}$, it is associated with the true causal graph. In this context, $\mathcal{G}^{global}_{ij}$ can be viewed as an indicator vector, where $\mathcal{G}^{global}_{ij} = 1$ signifies that node $i$ is the parent node of node $j$, indicating that node $j$ is influenced by node $i$. In contrast, $\mathcal{G}^{global}_{ij} = 0$ implies that node $j$ and node $i$ are unrelated. The global causal graph is primarily used to capture dependencies among confounders without focusing on the strength of the dependencies between any two dependent confounders. \par
Therefore, we only require a binary adjacency matrix to meet this need. We begin by initializing an all-one learnable adjacency matrix $\mathcal{G}^{global}$, which is subsequently binarized to meet the specified requirements. To make the binary operation continuous, we leverage the $Gumbel$-$Softmax$ to get the binary adjacency matrix, which gives a continuous approximation to sampling from the categorical distribution \cite{other1,other2,mrcrid}. We adopt a similar approach by adding Gumbel noise to the sigmoid function, which we formula as $Gumbel\text{-}Sigmoid$:
\begin{equation}
    \begin{aligned}
        Gumbel&\text{-}Sigmoid(\mathcal{G}^{global}) = \frac{exp((\mathcal{G}^{global} + \dot{g}) / \tau)}{exp((\mathcal{G}^{global} + \dot{g}) / \tau) + exp(\ddot{g} / \tau)},\\
    \end{aligned}
    \label{eq:4}
\end{equation}
where $\dot{g}$ and $\ddot{g}$ are two independent Gumbel noises, and $\tau \in (0, \infty)$ is a temperature parameter. As $\tau$ diminishes to zero, a sample from the $Gumbel\text{-}Sigmoid$ distribution becomes cold and resembles the one-hot samples. Our experiments show that a small fixed $\tau$ (e.g., 0.2) works well.

\begin{figure}[!tb]
    \centering
    \includegraphics[width = 0.8\linewidth]{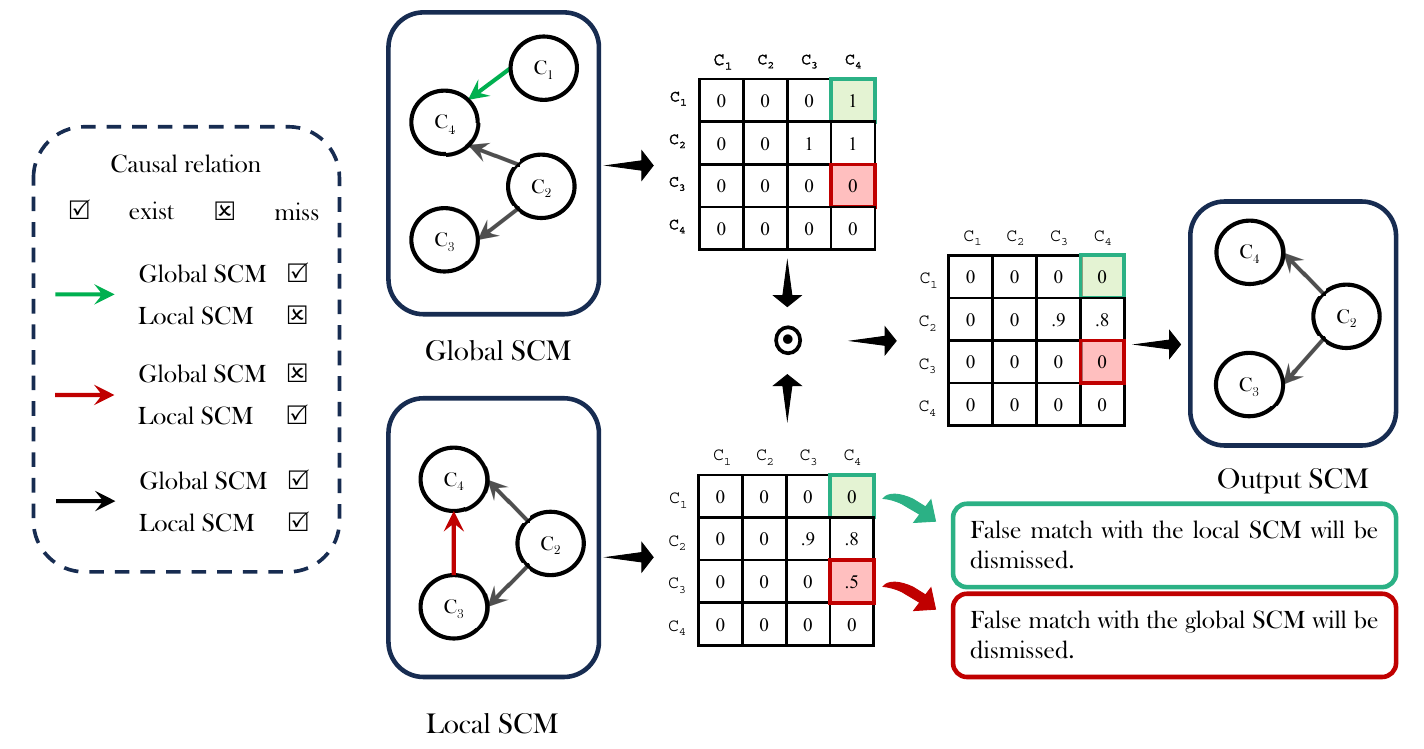}
    \caption{Example of a Global and Local Causal Structure.}
    \label{fig:4}
\end{figure}
\subsubsection{\textbf{Local SCM}} Once we obtain a usable global causal graph, we use local causal graphs to measure the strength of dependencies between confounders. The VAE-based models usually obtain the user preference representation in latent space from the user feedback data by encoders, and our model follows this process to obtain the mixed preference in feedback data:
\begin{equation}
    \begin{aligned}
        z = \mathsf{Encoder}(x_u),
    \end{aligned}
\end{equation}
where $z$ is the representation of mixed preference in feedback data. To separate the confounders from $z$, we use a nonlinear function to obtain the projection of $z$ in the space of confounders, which can be formulated as follows:
\begin{equation}
    \begin{aligned}
        c_i = \mathit{f_i}(z), \quad i\in[0,k],
    \end{aligned}
\end{equation}
where $c_i$ is the $i$-$th$ confounder and $\mathit{f_i}(\cdot)$ is the corresponding function. In practice, we use two linear layers to simulate. Then we can calculate the exogenous variables $\epsilon$ by transforming Eq. \ref{eq:3}:
\begin{equation}
    \begin{aligned}
        \epsilon = \{I - (\mathcal{G}^{global})^{\top}\}\mathbf{\mathrm{C}}.
    \end{aligned}
\end{equation}

In our model, the mapping function of confounders is shared among users. Therefore, to capture user specificity, we must introduce additional user-specific information to distinguish between users. The conventional approach involves using additional user embeddings as personalized user information. However, this leads to increased computational complexity and parameters in the model. To mitigate this drawback, we utilize a $\mathbf{\mathsf{MLP\ layer}}$ to map the encoder output to obtain personalized user preference without significantly increasing the temporal and spatial complexity of the model,
\begin{equation}
    \begin{aligned}
        \mathit{sInfo}_{u} = \mathsf{MLP}(z),
    \end{aligned}
\end{equation}
where $\mathit{sInfo_u}$ is the user-specific preference after MLP Layer, as Figure \ref{fig:5} shows. By incorporating additional user-specific information, we can obtain distinct representations of confounders for each user. Furthermore, we utilize an attention mechanism to calculate the strength of dependencies between confounders. In our experiments, we observed that the multi-head attention performs better as it effectively captures the heterogeneity among different confounders.
\begin{equation}
    \begin{aligned}
        \mathcal{G}^{local} = \mathsf{MultiHead}(\mathbf{\epsilon}, \mathit{sInfo_{u}}),
    \end{aligned}
\end{equation}
where the number of attention heads corresponds to the number of confounders. The $Multi$-$Head$ layer is illustrated in Figure \ref{fig:5}.

\subsubsection{\textbf{Causal Layer}} Given the global causal graph and local causal graphs, we perform calculations in the causal layer to obtain the final user-specific causal graph.
\begin{equation}
    \begin{aligned}
        \mathcal{G}^{u} = \mathcal{G}^{global} \odot \mathcal{G}^{local},
    \end{aligned}
\end{equation}
where $\odot$ is the element-wise multiplication. As Figure \ref{fig:4} shows, $\mathcal{G}^{u}_{ij} = 1$ if and only if both $\mathbf{\mathcal{G}^{global}_{ij} = 1}$ and $\mathcal{G}^{local}_{ij} \neq 0$ hold, which reveals two reasonable potential conditions. Firstly, the local causal graphs must adhere to the global causal graph. This means that any two confounders without a causal relationship in the global causal graph cannot establish causality through the local causal graphs. The global causal graph represents the true causal graph; thus, any causal relationship absent in the global causal graph, even if present in the local causal graphs, is considered erroneous and disregarded. Secondly, any two confounders without a causal relationship in the local causal graphs cannot influence the current user through the global causal graph. The local causal graphs capture the influence of confounders on the user. If there is no causal relationship between two confounders in the local causal graph, this causal pathway cannot influence the user. Hence, including such pathways would affect the final performance and is therefore not considered. Once we obtain the final causal graph, we can derive reconstructed causal representations of the confounders according to Eq \ref{eq:3}:
\begin{equation}
    \begin{aligned}
        \hat{c}_{i} &= g_{i}(\mathcal{G}^{u}_{i} \odot \{I - (\mathcal{G}^{global})^{\top}\}^{-1}\epsilon),\\
        \mathbf{\hat{C}} &= \{\hat{c}_{1}, \hat{c}_{2},...,\hat{c}_{k}\},
    \end{aligned}
    \label{eq:8}
\end{equation}
where $g_{i}(\cdot)$ is a mild nonlinear function \textcolor{black}{(\(sigmoid(\cdot)\) in our experiments), which is less sensitive to input variations, thereby facilitating more efficient parameter updates during the training process}. For any confounder $c_{i}$, $\odot$ represents considering only the influence of its parent nodes, excluding the influence of other irrelevant nodes. 
\subsubsection{\textbf{Mask Layer}} "do-operation" is a commonly used tool in causal theory through which you can tell us something about counterfactual problems, for example, `Will users prefer ice cream if the weather is warm?'. This mask layer can implement the "do-operation." We only need to provide an additional mask $\mathcal{G}^{mask}$, where $\mathcal{G}^{mask}_{ij} = 0$ indicates excluding the influence of node $i$ on node $j$ in the do-operation. Given the mask graph $\mathcal{G}^{mask}$, the reconstructed causal representations can derived from:
\begin{equation}
    \begin{aligned}
        &\mathcal{G}^{u\mbox{-}masked} = \mathcal{G}^{u} \odot \mathcal{G}^{mask},\\
        &\hat{c}_{i} = g_{i}(\mathcal{G}^{u\mbox{-}masked}_{i} \odot \{I - (\mathcal{G}^{global} )^{\top}\}^{-1}\epsilon).\\
    \end{aligned}
    \label{eq:8}
\end{equation}
If no explicit $\mathcal{G}^{mask}$ as input, $\mathcal{G}^{mask}$ is set as an all-one matrix to enhance the causal relation flow between confounders.
\subsubsection{\textbf{Identification of the Learned Graph}} As shown in Figure \ref{fig:5}, we will utilize a parametric model like Variational Autoencoder (VAE), combined with a $k \times k$ binary adjacency matrix, to fit the observed data. Unsupervised learning of the model might be infeasible due to the identifiability issue as discussed in \cite{ref01,ref02,ref03}. To demonstrate the identifiability of the learned graph, we prove that under appropriate conditions, the computation described above can lead to the recognition of the hypergraph of the true graph. Consider a marginal distribution $P(\textbf{C})$ induced by a Structural Equation Model (SEM) defined in Eq \ref{eq:3} with Directed Acyclic Graph (DAG) $\mathcal{G}$, and our SEM Eq \ref{eq:8} induces the same marginal distribution, where the binary adjacency matrix represents a DAG $\mathcal{H}$, we can obtain Lemma \ref{lamma:1} if the $g_{i}(\cdot)$ is not a constant function and its proof.
\begin{lamma}
    $\mathcal{H}$ is a super\text{-}graph of $\mathcal{G}$, i.e., all the edges in $\mathcal{G}$ also exist in $\mathcal{H}.$
    \label{lamma:1}
\end{lamma}
\begin{proof}
First, let's consider the case where the $g_{i}(\cdot)$ is a constant, w.r.t. $c_{j}$ whether the $\textbf{A}_{ji} = 0 \enspace\text{or} \enspace 1$ do not affect $c_{i}$, but will change the causal graph $\mathcal{H}$, so when $g_{i}(\cdot)$ is constant, we can not uniquely identified the $\mathcal{H}$ from $P(\textbf{C})$. Fortunately, in recommender systems, $g_{i}(\cdot)$ usually satisfies the non-constant condition. Then, we restrict $g_{i}$ to be non-constant, w.r.t. all $c_{j}$, $j \neq i$ to meet the causal minimality condition. \par
It suffices to show that if $c_{j}$ is not a parent of $c_{i}$ in $\mathcal{H}$, then $c_{j}$ is not a parent of $c_{i}$ in $\mathcal{G}$, either. That $c_{j}$ is not a parent of $c_{i}$ in $\mathcal{H}$ indicates $\textbf{A}_{ji} = 0$. Therefore, $g_{i}(\textbf{A}_{i} \odot \textbf{C})$ is a constant function w.r.t. $c_{j}$. For the reduced SEM with functions $g_{i}$’s and causal DAG $\mathcal{G}$, we conclude that $c_{j} \notin c_{pa_{i}}$ and the input arguments of $g_{i}$ do not contain. Thus, $c_{j}$ cannot be a parent of $c_{i}$ in $\mathcal{G}$.
\end{proof}
As  Peters' book \cite{peters14a}[Theorem 27] shows, if the $P(\textbf{C})$ is generated by a restricted additive noise model (ANM), the true causal graph is identifiable. Thus, we further assume a restricted ANM for the data-generating procedure to ensure the true causal graph $\mathcal{G}$ is identifiable. We then obtain the following proposition with proof.
\begin{proposition}
    Assume a restricted ANM with graph $\mathcal{G}$ and distribution $P(\textbf{C})$ so that the original SEM is identifiable. If the parameterized SEM in the form of Eq. \ref{eq:8} with graph $\mathcal{H}$ induces the same $P(\textbf{C})$, then $\mathcal{H}$ is a super-graph of $\mathcal{G}$.
\end{proposition}
\begin{proof}
    Recall that the reduced SEM with $g_{i}$’s and graph $\mathcal{G}$ satisfies the causal minimality condition and has the same distribution $P(\textbf{C})$. With the identifiability result of restricted ANMs \cite{peters14a}, we know that $\mathcal{G}$ is identical. Applying Lemma \ref{lamma:1} completes the proof.
\end{proof}
Then, we can apply a parametric model and a binary adjacency matrix to fit the SEM in Eq. \ref{eq:3}. Suppose the causal relationships fall into the chosen model functions, and we can obtain the exact solution that minimizes the negative log-likelihood given infinite samples. In that case, the resulting SEM has the same distribution \cite{ref04}. Consequently, we obtain an acyclic supergraph from which existing nonlinear variable selection methods can be used to learn the parental sets and the causal graph.

\subsubsection{\textbf{Mix Layer}} With the reconstructed causal representation of confounders \(\mathbf{\hat{C}}\) and user-specific preference, we can drive the mixed preference representation of feedback data in latent space. For \(k\) confounders, a user may be influenced by some rather than all. To retain this characteristic, we take advantage of the attention mechanism to model the confounder's influence on the user. Specifically, we calculate the attention score as follows:
\begin{equation}
    \begin{aligned}
        \mathit{Q} = \mathit{f}(\mathsf{Norm}(\mathit{sInfo_{u}})),\quad\mathit{K} &= \mathit{g}(\mathsf{Norm}(\mathbf{\hat{C}})), \quad\mathit{V} = \mathit{h}(\mathbf{\hat{C}}),\\
        score &= \mathsf{Softmax}(\frac{\mathit{Q}\mathit{K}^{\top}}{\sqrt{d}}),\\
        \mathbf{{C}^{u}} &= score \cdot \mathit{V},
    \end{aligned}
\end{equation}
where \(\mathit{f}\), \(\mathit{g}\) and \(\mathit{h}\) are learnable linear layers, \(\mathsf{Norm}(\cdot)\) means normalization, \(d\) is the latent embedding size, and \(\mathbf{{C}^{u}}\) is the representation of confounders influence on user \(u\). Then we mix the \({sInfo_{u}}\) and \(\mathbf{{C}^{u}}\) to get the reconstructed mixed user preference representation of feedback data in latent space as follows:
\begin{equation}
    \begin{aligned}
        \hat{z} = \mathsf{FFN}(\mathit{sInfo_{u}} +  \mathbf{{C}^{u}}),
    \end{aligned}
\end{equation}
where \(\mathsf{FFN}(\cdot)\) denotes the feed-forward layer, but in practice, we found simply add \(\mathit{sInfo_{u}}\) and  \(\mathbf{{C}^{u}}\) is enough. The Mix layer combines the output of user-specific preferences with the output of the causal layer - the causal representations of confounders, resulting in a user representation under the influence of confounders. This user representation is then input into the decoder to reconstruct observed data.\par
\begin{figure}[!tb]
    \centering
    \includegraphics[width= 0.9\linewidth]{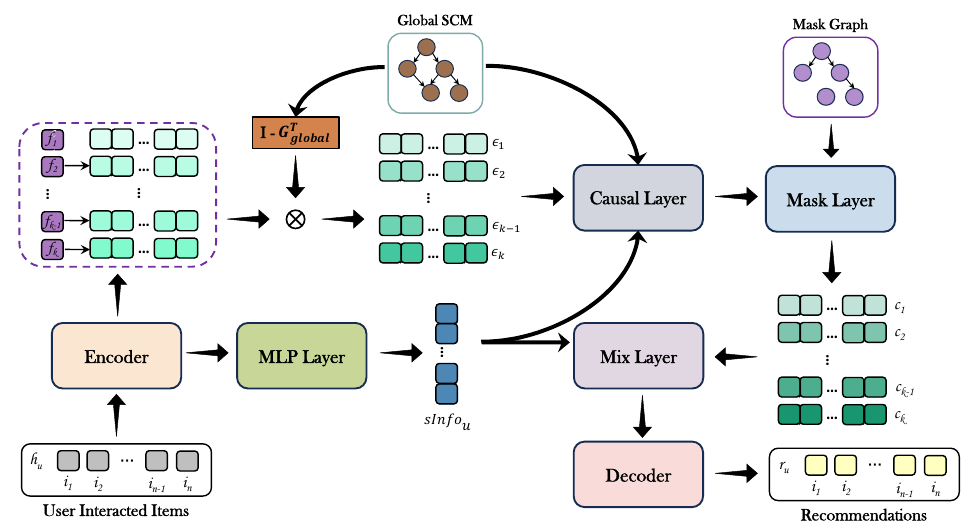}
    \caption{The architecture of CSA-VAE.}
    \label{fig:5}
\end{figure}
\section{Inference and Learning Strategy}
The overall architecture of our model is depicted in Figure \ref{fig:5}. This section describes the inference stage to make recommendations and the training stage to train our model to learn user preferences and causal graphs simultaneously.\par
\subsection{Inference process}
Our model CSA-VAE can explicitly model the unobserved confounders and user-specific preference, thus choosing whether to incorporate the influence of confounders.
\subsubsection{\textbf{Recommendation with confounders}} Given the reconstructed mixed user preference representation \(\hat{z}\), we can get the predicted user-item relevance score as follows:
\begin{equation}
    \begin{aligned}
        \{\hat{x_1},\hat{x_2},...,\hat{x_n}\} = \mathsf{Decoder}(\hat{z}).
    \end{aligned}
\end{equation}
\subsubsection{\textbf{Recommendation without confounders}} We use \(\mathit{Mask Layer}\) to completely Shielding from the influence of confounders with the all-zero \(\mathcal{G}^{mask}\), the simplified formula is as follows:
\begin{equation}
    \begin{aligned}
        \{\hat{x_1},\hat{x_2},...,\hat{x_n}\} = \mathsf{Decoder}(\mathsf{FFN}(\mathit{sInfo_{u}})).
    \end{aligned}
\end{equation}
\subsubsection{\textbf{User Controllable Recommendation}}
\textcolor{black}{The causal graph and the learned representations of the confounders allow users to interactively adjust the representation of preferences of recommender systems learned from their past behavior to better align with their current preferences. In this paper, we formalize the task of user-controllable recommendation.}\par
\textbf{Task Definition of User Controllable Recommendation.}  \textcolor{black}{Let \( x_u = \{ x_{u, 1}, x_{u, 2}, \dots, x_{u, n} \} \) represent the historical interaction data of user \( u \), where each \( x_{u, i} \) corresponds to a specific past interaction between the user and an item. Let \( z \) be the learned mixed preference representation of the user derived from their past behavior, and \( c_u \) represent the latent confounders that may influence the user's preferences. The user-controllable recommendation task is designed in the recommendation model to allow users to adjust the learned preference representation \( z \) to align their current intentions or preferences.}\par
\textbf{Solution based on CSA-VAE.} \textcolor{black}{Our model CSA-VAE incorporates confounders \( c_u \) that may capture external influences on users' preferences. We model confounders as latent variables, and the user can adjust their influence interactively. Users can modify the confounder representations \( c_u \) to better align with their current preferences as follows:}
\begin{equation}
    p_{\theta}(x_u) = E_{p_{\theta}(c_{u})} \left[ \int p_{\theta}(x_{u} | z_{u}, do(c_{u}^{i}), c_{u}^{j})p_{\theta}(z_{u})dz_{u}\right], i,j \in [1,k] \enspace and\enspace i \neq j,
\end{equation}
 \textcolor{black}{where \(do(\cdot)\) represents a causal intervention, meaning that the confounders \( c_u^i \) are manipulated or set to specific values during the recommendation process. ``do-operation'' is a key feature of causal inference, allows us to model the influence of changing these confounders on the user's preferences. The ability to manipulate confounders in this way gives the user control over how these confounders influence their current mixed preferences and the generated recommendations.}
\subsection{Training process}
\subsubsection{\textbf{Evidence Lower Bound}} Once we obtain the latent representation \(\hat{z}\) of users' mixed preferences in the feedback data, we can reconstruct the observed data as follows:
\begin{equation}
    \begin{aligned}
        p_{\theta}(x_{u} | \hat{z}) = p_{\theta}(x_{u}, \hat{z} | \mathit{sInfo_u}, \textbf{C}).
    \end{aligned}
\end{equation}
We follow the variational autoencoder (VAE) paradigm \cite{vae} and optimize $\theta$ by maximizing the lower bound $\sum_u \ln p_\theta(x_u)$, where $\ln p_\theta(x_u)$ is bounded as follows:
\begin{equation}
    \begin{aligned}
        \ln p_{\theta}(x_{u}) \geq E_{q(z,|x_u,\mathit{sInfo_u},\textbf{C})}\left[\ln p(x_{u} | z, \mathit{sInfo_u}, \textbf{C})]\right. - D_{KL}(q(z|x_u,\mathit{sInfo_u}, \textbf{C}) \| p(z|\mathit{sInfo_u},\textbf{C})).\\
    \end{aligned}
\end{equation}
The relevant proofs are provided in Appendix \ref{elbo}.
\subsubsection{\textbf{Constraints of the Causal Graph}} \textcolor{black}{Causal inference requires discovering the underlying causal structure between variables. A DAG is commonly used to represent such causal relationships because it naturally encodes the assumption that causality flows in one direction and does not form loops. Thus, the causal adjacency matrix $\mathbf{A}$ is constrained to be a DAG. We employ a continuous distinguishable constraint function instead of the traditional combinatorial DAG constraint \cite{dag}. This function attains zero if, and only if the adjacency matrix $\mathbf{A}$ corresponds to a DAG \cite{dag}:}
\begin{equation}
    \begin{aligned}
        \mathbf{H}(\mathbf{A}) = tr((I + \frac{c}{k}\mathbf{A} \odot \mathbf{A})^{k}) - k = 0,
    \end{aligned}
    \label{eq:19}
\end{equation}
\textcolor{black}{where c is an arbitrary positive number, controls the spectral radius of the graph, and ensures that the adjacency matrix does not lead to significant or unstable eigenvalues. $k$ is the number of confounders. The value of $c$ is the spectral radius of \textbf{A}, and due to nonnegativity, it is bounded by the maximum row sum by the Perron-Frobenius theorem.}

\subsubsection{\textbf{Diversity constraint of confounders}} To ensure the diversity and comprehensiveness of confounders, we incorporate a similarity penalty in the loss function to guide the diversity of parameters. This method effectively prevents the model from focusing too narrowly on specific patterns, thereby enhancing its generalization capability. By employing this penalty mechanism, the model can explore a broader parameter space during training, resulting in more comprehensive and diverse representations of confounders. The formula is as follows:
\begin{equation}
\centering
    \begin{aligned}
        \mathcal{L}_{sim} &= \sum^{k}_{i}\sum^{k}_{j}Cosine\mbox{-}Similarity(\mathsf{Norm}(\epsilon_i),\mathsf{Norm}(\epsilon_j)) \quad (i \neq j).\\
    \end{aligned}
    \label{eq:20}
\end{equation}
\subsubsection{\textbf{Objective function}} The training procedure of our model reduces to the following constrained optimization:
\begin{equation}
    \begin{aligned}
        maximize& \quad \ln p_{\theta}(x_{u}) \geq E_{p_{\theta}(\textbf{C)}}\left\{ E_{q_{\theta}(z_{u} | x_{u}, \textbf{C})}[\ln p_{\theta}(x_{u} | z_{u}, \textbf{C})]\right.\\
        &- D_{KL}(q_{\theta}(z_{u} | x_{u}, \textbf{C}) \| p_{\theta}(z_{u}))\left\}\right.,\\
        s.t.& \quad (\ref{eq:19}),\\
        s.t.& \quad (\ref{eq:20}).
    \end{aligned}
\end{equation}
By the lagrangian multiplier method, we have the new loss function:
\begin{equation}
    \begin{aligned}
        \mathcal{L} = -\textbf{ELBO} + \textbf{H}(\mathcal{G}^{global}) + \textbf{H}(\mathcal{G}^{local}) + \mathcal{L}_{sim}.
    \end{aligned}
\end{equation}
\section{Experiments}
In this section, we present the extensive experiments conducted on one semi-simulated dataset and five real-world datasets to demonstrate the effectiveness of the proposed CSA-VAE, with an emphasis on answering the following research questions:\par
\begin{itemize}
    \item \textbf{RQ 1}: Can CSA-VAE obtain a useful causal graph of confounders without relevant labels? How does the strength of the causal relationship obtained compare to the true value?\par
    \item \textbf{RQ 2}: Can CSA-VAE achieve better performance compared to other baselines? How about the performance of CSA-VAE using only user preferences for recommendations? \par
    \item \textbf{RQ 3}: How do the causal relationships of confounders enhance the model's performance?\par
    \item \textbf{RQ 4}: How do the number of confounders influence the performance of CSA-VAE? Is the greater the number of confounders, the better? \par
\end{itemize}
\subsection{Dataset}
It is difficult to verify the effectiveness of the causal graphs learned by CSA-VAE as the information of confounders is unobserved in the real dataset. Thus, we first validate the effectiveness of causal graphs on synthetic datasets and subsequently evaluate the recommendation performance of CSA-VAE on real-world datasets.
\begin{table}
    \centering
    \resizebox{0.7\linewidth}{!}{%
\begin{tabular}{@{}llrrrc@{}}
\toprule
Dataset  & Type                           & \#Interactions & \#User  & \#Items & Sparsity \\ \midrule
Coat     & \multirow{4}{*}{Full-observed} & 11,600         & 290     & 300     & 86.67\%    \\
Reasoner &                                & 58,497         & 2,997   & 4,672   & 99.58\%    \\
Yahoo!R3 &                                & 365,704        & 15,400  & 1,000   & 97.63\%    \\
\rowcolor{white}Kuairec  &                                & 12,530,806     & 7,176   & 10,728  & 83.72\%    \\
\midrule
Epinions & \multirow{6}{*}{Normal}        & 188,478        & 116,260 & 41,269  & 99.99\%    \\
ML-100k  &                                & 100,000        & 943     & 1,682   & 93.69\%    \\
ML-1M    &                                & 1,000,209      & 6,040   & 3,706   & 95.54\%    \\
ML-10M   &                                & 10,000,054     & 69,878  & 10,677  & 98.66\%    \\
ML-20M   &                                & 20,000,263     & 138,493 & 26,744  & 99.46\%    \\
ML-25M   &                                & 25,000,096     & 162,542 & 59,048  & 99.74\%    \\ \bottomrule
\end{tabular}}
    \caption{Statistics of the datasets}
    \label{tab:1}
\end{table}
\subsubsection{\textbf{Synthetic Dataset.}} We conduct experiments on synthetic data generated by the following process: We first assume that users are influenced by four confounders, where the exogenous variables for each confounder are generated by sampling from Gaussian distributions with mean and variance sampled from uniform distributions $[-3, 3]$ and $[0.01, 4]$, respectively. The intrinsic preferences of users are sampled from a standard normal distribution. Given the user's preference value $\mathcal{U}$, we obtain the user's personalized weights $w$ by sampling from a Poisson distribution. We generate samples using the causal structure model as shown in Eq \ref{eq:2}. Finally, we input the confounders and user preferences into a two-layer MLP to generate the final observed value $\mathcal{X}$. Additional details and a formal description can be found in the Appendix \ref{sydata}.
\subsubsection{\textbf{Real-World datasets}} To comprehensively and fairly validate the effectiveness of the model, we conducted experiments using nine publicly available datasets that encompass a variety of recommendation scenarios (such as movies and clothes) and different densities. \textcolor{black}{\dataset{Coat}, \dataset{Yahoo R3}, \dataset{Reasoner}\footnote{\href{https://reasoner2023.github.io/}{https://reasoner2023.github.io/}} \cite{reasoner}, and \dataset{Kuairec}\footnote{\href{https://kuairec.com/}{https://kuairec.com/}} \cite{kuairec} have fully observed data as the test set with the ground-truth relevance information. Following prior works, we binarize the ratings in \dataset{Yahoo R3} and \dataset{Coat} by setting ratings $\geq$ 4 to 1 and the rest to 0. For \dataset{Reasoner}, we set ratings $\geq$ 4 to 1 and the rest to 0 as the regular train-test set and use the true user preference label like-unlike as the test set for evaluating the performance of models in capturing real user preference. For \dataset{Kuairec}, we use the sparse dataset for the train and the dense dataset for the test; the rating is binarized based on the ratio of user watching ratio. Specifically, setting watching ratio $\geq$ 2 to 1 and the rest to 0.} We select five datasets of varying sizes ranging from 100k to 25M: \dataset{ML}-\dataset{100K}, \dataset{ML}-\dataset{1M}, \dataset{ML}-\dataset{10M}, \dataset{ML}-\dataset{20M} and \dataset{ML}-\dataset{25M} collected from the MovieLens website\footnote{\href{https://grouplens.org/datasets/movielens/}{https://grouplens.org/datasets/movielens/}} to validate the robustness of the model to the dataset size. Additionally, we leverage the \dataset{Epinions} dataset, which originates from Epinions.com, a website where users can write reviews on various products and services and also rate the reviews written by other users. Following prior works, \cite{heng1,heng2}, we remove the "inactive" users who interact with fewer than 20 items and the "unpopular" items who have interacted with users less than 10 times. We split the dataset into 70\% for training, 20\% for testing, and the remaining for validation. All user ratings greater than or equal to four are set to 1, while the rest are set to 0.
\subsubsection{\textbf{Baselines}} We compare our method with the corresponding base models and the state-of-the-art de-confounding methods that can alleviate the confounding bias in recommender systems in the presence of unobserved confounders. 
\begin{itemize}
    \item \textbf{MF} \cite{mf}: MF is a popular technique used in recommendation systems to predict user preferences for items. It is particularly effective for collaborative filtering
    \item \textbf{Multi-VAE} \cite{mutivae}: Variational autoencoders (VAEs) to collaborative filtering for implicit feedback with VAE.
    \item \textbf{Muti-DAE} \cite{mutivae}:  variational autoencoders (VAEs) to collaborative filtering for implicit feedback with DAE. 
    \item \textbf{Macrid-VAE} \cite{mrcrid}: Achieves macro disentanglement by inferring the high-level concepts associated with user intentions while simultaneously capturing a user's preference regarding the different concepts. 
    \item \textbf{Rec-VAE} \cite{recvae}: RecVAE introduces several novel ideas to improve Mult-VAE. 
    \item \textbf{CDAE} \cite{cdae}: A novel method for top-N recommendation that utilizes the idea of Denoising Auto-Encoders. 
    \item \textbf{InvPref} \cite{invpref}: InvPref assumes the existence of multiple environments as proxies of unmeasured confounders and applies invariant learning to learn the user’s invariant preference. 
    \item \textbf{IDCF} \cite{icdf}: A general de-confounded recommendation framework that applies proximal causal inference to infer the unmeasured confounders and identify the counterfactual feedback with theoretical guarantees.
\end{itemize}
CSA-VAE focuses on learning a causal graph without the use of relevant labels. To ensure a fair comparison, we do not include models that incorporate additional information (e.g., movie names, categories), such as models SEM-MacridVAE \cite{sem1}, CaD-VAE \cite{sem2} and PlanRec \cite{sem3}.
\subsection{Experimental Settings}
\subsubsection{\textbf{Setups.}} We implement CSA-VAE and baselines in PyTorch. All models are trained with the Adam optimizer via early stopping at patience = 10. We set the learning rate to 1e-3 and the $l_2$-regularization weight to 1e-6. For CSA-VAE, we tune the hyper-parameter concepts $k$ in the range of $[1,2,4,8,16,32]$ for different datasets. To detect significant differences in CSA-VAE and the best baseline on each dataset, we repeated their experiments five times by varying the random seeds. We choose the average performance to report. All ranking metrics are computed at a cutoff K = $[10,30]$ for the Top-$k$ recommendation. Our implementation of the baselines is based on the original paper or the open codebase Recbole \cite{recbole[2.0]}.\par
\subsubsection{\textbf{Evaluation Metrics.}} Note that the sampling-based evaluation approach does not truly reflect the ability of the model to capture the true preferences of users. Simply fitting the data may also have better performance. To this end, we report the all-ranking performance w.r.t. two widely used metrics: Recall and NDCG cut at K = $[10, 30]$. To measure the popularity of recommended items, we use average popularity rank (AVP) as the other indicator, and the formula is:
\begin{equation}
    \begin{aligned}
        \mathrm{Average\ Popularity\ (AVP)@K}=\frac{1}{|U|} \sum_{u \in U } \frac{\sum_{i \in R_{u}} \phi(i)}{|R_{u}|}, \nonumber
    \end{aligned}
\end{equation}
where $\phi(i)$ is the ascending order of times item $i$ has been rated in the training set, $\mathbf{R}_{u}$ is the recommended list of items for user $u$.
\begin{figure}
    \centering
    \includegraphics[width=0.9\linewidth]{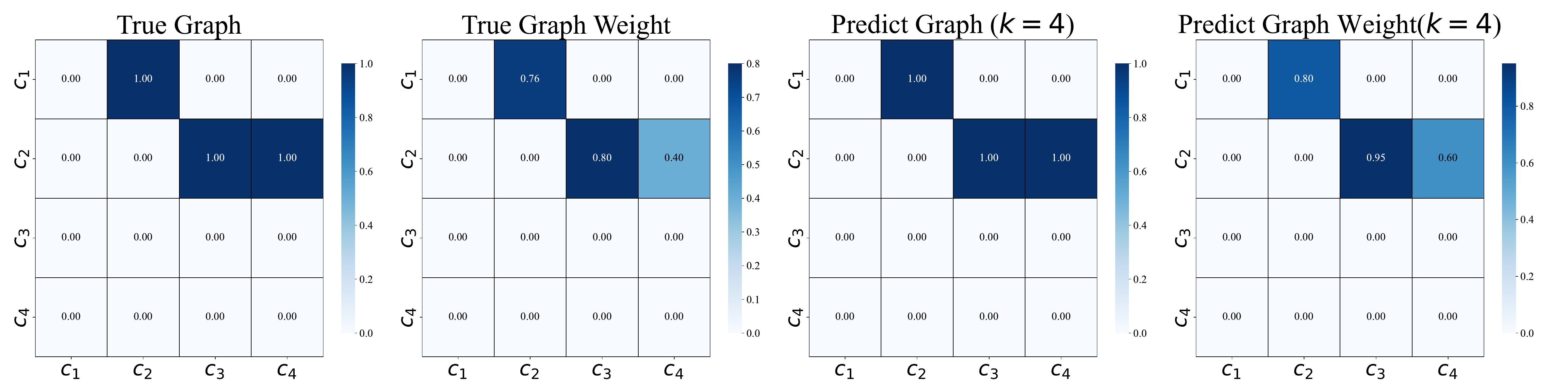}
    \caption{In the graph prediction experiment on the synthetic data, the true graph (left) and the predicted graph (right).}
    \label{fig:6}
\end{figure}
\begin{table}[!t]
\centering
\renewcommand{\arraystretch}{1.2}
\resizebox{\linewidth}{!}{%
\begin{tabular}{@{}l|c|c|cccccccccc@{}}
\toprule
Datasets                  & Metric                  & K  & MF      & CDAE    & Multi-DAE     & Multi-VAE     & Macrid-VAE & Rec-VAE       & InvPref & ICDF    & CSA-VAE          & Imp.(\%) \\ \midrule
\multirow{4}{*}{Coat}     & \multirow{2}{*}{Recall \(\uparrow\)} & 10 & 0.03241 & 0.03418 & 0.03530       & 0.03461       & 0.03328    & {\ul 0.03539} & 0.03332 & 0.03439 & \textbf{0.03664} & 3.81\%   \\
                          &                         & 30 & 0.09966 & 0.10180 & 0.09815       & 0.09909       & 0.09810    & {\ul 0.10330} & 0.09806 & 0.09684 & \textbf{0.10647} & 3.07\%   \\
                          & \multirow{2}{*}{NDCG \(\uparrow\)}   & 10 & 0.05126 & 0.05362 & {\ul 0.05834} & 0.05534       & 0.05256    & 0.05665       & 0.05184 & 0.05022 & \textbf{0.05916} & 1.41\%   \\
                          &                         & 30 & 0.07871 & 0.08073 & 0.08212       & 0.08013       & 0.07846    & {\ul 0.08324} & 0.07776 & 0.07475 & \textbf{0.08557} & 2.79\%   \\ \midrule
\multirow{4}{*}{Yahoo!R3} & \multirow{2}{*}{Recall \(\uparrow\)} & 10 & 0.03225 & 0.04493 & 0.06424       & {\ul 0.06560} & 0.02066    & 0.04648       & 0.02548 & 0.02519 & \textbf{0.06768} & 3.17\%   \\
                          &                         & 30 & 0.07687 & 0.10290 & 0.15280       & {\ul 0.15310} & 0.04693    & 0.10200       & 0.06069 & 0.05899 & \textbf{0.16169} & 5.61\%   \\
                          & \multirow{2}{*}{NDCG \(\uparrow\)}   & 10 & 0.01707 & 0.02440 & 0.03069       & {\ul 0.03149} & 0.01121    & 0.02503       & 0.01330 & 0.01372 & \textbf{0.03480} & 10.51\%  \\
                          &                         & 30 & 0.02893 & 0.03996 & 0.05387       & {\ul 0.05447} & 0.01846    & 0.03999       & 0.02275 & 0.02265 & \textbf{0.05928} & 8.83\%   \\ \midrule
\multirow{4}{*}{Reasoner} & \multirow{2}{*}{Recall \(\uparrow\)} & 10 & 0.00338 & 0.00234 & 0.00268       & 0.00356       & 0.00194    & {\ul 0.00386} & 0.00277 & 0.00276 & \textbf{0.00508} & 31.61\%  \\
                          &                         & 30 & 0.00959 & 0.00741 & 0.00961       & {\ul 0.01136} & 0.00723    & 0.00978       & 0.00926 & 0.00972 & \textbf{0.01445} & 27.20\%  \\
                          & \multirow{2}{*}{NDCG \(\uparrow\)}   & 10 & 0.00191 & 0.00129 & 0.00152       & 0.00165       & 0.00111    & {\ul 0.00193} & 0.00157 & 0.00155 & \textbf{0.00274} & 41.97\%  \\
                          &                         & 30 & 0.00366 & 0.00273 & 0.00343       & {\ul 0.00381} & 0.00261    & 0.00355       & 0.00339 & 0.00333 & \textbf{0.00488} & 28.08\%  \\ \midrule 
                          &                                      & 10 & {\ul 0.07115} & 0.06744       & 0.06826       & 0.06296       & 0.06382    & 0.06570       & 0.06939       & 0.06672 & \textbf{0.07442} & 4.59\%   \\
                          &\multirow{-2}{*}{Recall \(\uparrow\)}  & 30 & 0.09671       & 0.12428       & {\ul 0.12864} & 0.12455       & 0.12306    & 0.12640       & 0.12489       & 0.11964 & \textbf{0.13104} & 1.87\%   \\
                          &                                      & 10 & 0.44834       & 0.43268       & 0.45814       & 0.43918       & 0.44796    & 0.44110       & {\ul 0.46040} & 0.43738 & \textbf{0.48594} & 5.55\%   \\
\multirow{-4}{*}{KuaiRec}  &\multirow{-2}{*}{NDCG \(\uparrow\)}    & 30 & 0.27466       & {\ul 0.35317} & 0.36334       & 0.35451       & 0.35196    & 0.35164       & 0.34062       & 0.32359 & \textbf{0.36743} & 4.03\%   \\ \bottomrule
\end{tabular}
}
\caption{The overall performance comparison results of applying our model and baselines on four real-world full-observed datasets. We evaluated the recommendation performance as a ranking task, underlined the best baseline result in each line, and put the best result in each line in bold; Higher Recall and NDCG mean better model performance. The arrow `\(\uparrow\)' (or `\(\downarrow\)') denotes that the higher (or lower) value means better performance on the metric. The `Imp.' row reports the relative improvement or decline of CSA-VAE against the best baseline. The result is calculated based on the mean of five repetitions with different random seeds for all models on each metric.}
\label{tab:3}
\end{table}
\subsection{Performance on the Synthetic Dataset (RQ1).}
CSA-VAE can obtain a useful causal graph of confounders without relevant labels. The synthetic data set we used contained four confounders, resulting in a power of $2^{k(k-1)}$ possible relationships. Although the number of categories is not extensive, it still poses a challenging task. By the causal relationships between confounders present in the synthetic data, we can unambiguously determine the ability of CSA-VAE to capture the causal relationships between confounders. Due to the strong correlation between the local graph and users, we only present the global graph obtained by the CSA-VAE here. As shown in Figure \ref{fig:6}, the global graph learned by CSA-VAE is well aligned with the ground truth graph, thus demonstrating the ability of CSA-VAE to effectively capture the causal relationships between confounders. It is important to emphasize that we used the $Gumbel\text{-}Sigmoid$ shown in Eq. \ref{eq:4}, resulting in an approximation of a binary causal graph by CSA-VAE. The final experimental results strongly support this approach.\par
Additionally, we found that the correlation strength between the confounders obtained by CSA-VAE is slightly higher than the true value. This occurs because the model treats the learning of causal relations as a binary classification task. Although the global causal graph is specialized for determining whether a causal relationship exists between two confounders, the local causal graph also plays a role in this function. To achieve binary classification, the model tends to exaggerate values to ensure accuracy, resulting in a slightly higher correlation strength between confounders with a causal relationship. A larger causal relationship helps the model better measure the impact of confounders on user preferences, thus aiding in more accurate modeling of confounders and user preferences.
\begin{table}[!t]
\centering
\renewcommand{\arraystretch}{1.2}
\resizebox{\linewidth}{!}{%
\begin{tabular}{@{}l|c|ccccccccccc@{}}
\toprule
Datasets                           & Metric                  & K  & \textbf{MF} & \textbf{CDAE} & \textbf{Multi-DAE} & \textbf{Multi-VAE} & \textbf{Macrid-VAE} & \textbf{Rec-VAE} & \textbf{InvPref} & \textbf{ICDF} & \textbf{CSA-VAE} & \textbf{Imp.(\%)} \\ \midrule
\multirow{4}{*}{\textbf{ML-100K}}  & \multirow{2}{*}{Recall \(\uparrow\)} & 10 & 0.02393     & 0.04412       & 0.05554            & {\ul 0.05620}      & 0.02259             & 0.05396          & 0.01344          & 0.01292       & \textbf{0.06207} & 10.44\%           \\
                                   &                         & 30 & 0.06294     & 0.09724       & {\ul 0.15466}      & 0.15226            & 0.05641             & 0.12149          & 0.03476          & 0.03425       & \textbf{0.15948} & 3.12\%            \\
                                   & \multirow{2}{*}{NDCG \(\uparrow\)}   & 10 & 0.03067     & 0.04571       & 0.05168            & 0.05200            & 0.02299             & {\ul 0.05588}    & 0.04518          & 0.03517       & \textbf{0.05876} & 5.15\%            \\
                                   &                         & 30 & 0.04358     & 0.06497       & {\ul 0.08830}      & 0.08715            & 0.03515             & 0.07879          & 0.03246          & 0.04245       & \textbf{0.09307} & 5.40\%            \\ \midrule
\multirow{4}{*}{\textbf{ML-1M}}    & \multirow{2}{*}{Recall \(\uparrow\)} & 10 & 0.02665     & 0.02825       & 0.02903            & 0.02843            & 0.2726              & 0.02739          & 0.02926          & {\ul 0.02997} & \textbf{0.03017} & 0.67\%            \\
                                   &                         & 30 & 0.06525     & 0.07540       & 0.07647            & {\ul 0.08782}      & 0.07729             & 0.07762          & 0.07316          & 0.06961       & \textbf{0.08906} & 1.41\%            \\
                                   & \multirow{2}{*}{NDCG \(\uparrow\)}   & 10 & 0.03373     & 0.03732       & 0.03054            & 0.02954            & 0.02372             & {\ul 0.03733}    & 0.03515          & 0.03425       & \textbf{0.03855} & 3.27\%            \\
                                   &                         & 30 & 0.05079     & 0.05497       & 0.05268            & 0.05235            & 0.05476             & {\ul 0.05593}    & 0.05503          & 0.05289       & \textbf{0.05589} & -0.07\%           \\ \midrule
\multirow{4}{*}{\textbf{ML-10M}}   & \multirow{2}{*}{Recall \(\uparrow\)} & 10 & 0.03101     & {\ul 0.03629} & 0.03552            & 0.03595            & 0.03565             & 0.3539           & 0.03482          & 0.03412       & \textbf{0.04785} & 31.85\%           \\
                                   &                         & 30 & 0.06686     & 0.08193       & {\ul 0.12127}      & 0.11968            & 0.11200             & 0.11933          & 0.07939          & 0.07917       & \textbf{0.13134} & 8.30\%            \\
                                   & \multirow{2}{*}{NDCG \(\uparrow\)}   & 10 & 0.03610     & {\ul 0.04001} & 0.03374            & 0.03410            & 0.03725             & 0.03250          & 0.03834          & 0.03866       & \textbf{0.04506} & 12.62\%           \\
                                   &                         & 30 & 0.04763     & 0.05579       & {\ul 0.06500}      & 0.06456            & 0.05780             & 0.06334          & 0.05346          & 0.05407       & \textbf{0.07544} & 16.06\%           \\ \midrule
\multirow{4}{*}{\textbf{ML-20M}}   & \multirow{2}{*}{Recall \(\uparrow\)} & 10 & 0.03462     & 0.03726       & 0.03802            & {\ul 0.03974}      & 0.03445             & 0.03660          & 0.03738          & 0.03714       & \textbf{0.04064} & 2.26\%            \\
                                   &                         & 30 & 0.07543     & 0.08559       & {\ul 0.12024}      & 0.11965            & 0.11412             & 0.11681          & 0.08554          & 0.08564       & \textbf{0.12260} & 1.96\%            \\
                                   & \multirow{2}{*}{NDCG \(\uparrow\)}   & 10 & 0.03891     & {\ul 0.04300} & 0.03746            & 0.03964            & 0.03631             & 0.03496          & 0.04083          & 0.04215       & \textbf{0.04534} & 5.44\%            \\
                                   &                         & 30 & 0.05232     & 0.05952       & 0.06803            & {\ul 0.06904}      & 0.06567             & 0.06537          & 0.05718          & 0.05677       & \textbf{0.07341} & 6.33\%            \\ \midrule
\multirow{4}{*}{\textbf{ML-25M}}   & \multirow{2}{*}{Recall \(\uparrow\)} & 10 & 0.02972     & 0.03194       & {\ul 0.03612}      & 0.03471            & 0.03293             & 0.03356          & 0.03374          & 0.03399       & \textbf{0.04298} & 18.99\%           \\
                                   &                         & 30 & 0.07140     & 0.07828       & {\ul 0.11434}      & 0.11129            & 0.09572             & 0.11106          & 0.07969          & 0.08177       & \textbf{0.12220} & 6.87\%            \\
                                   & \multirow{2}{*}{NDCG \(\uparrow\)}   & 10 & 0.03697     & {\ul 0.03869} & 0.03652            & 0.03458            & 0.03173             & 0.03332          & 0.03731          & 0.03551       & \textbf{0.04373} & 13.03\%           \\
                                   &                         & 30 & 0.05136     & 0.05493       & {\ul 0.06637}      & 0.06360            & 0.06553             & 0.06259          & 0.05296          & 0.05288       & \textbf{0.07331} & 10.46\%           \\ \midrule
\multirow{4}{*}{\textbf{Epinions}} & \multirow{2}{*}{Recall \(\uparrow\)} & 10 & 0.00897     & 0.01034       & {\ul 0.01765}      & 0.01732            & 0.01191             & 0.01650          & 0.01346          & 0.01311       & \textbf{0.01932} & 9.46\%            \\
                                   &                         & 30 & 0.01828     & 0.02505       & 0.03956            & 0.04067            & 0.01834             & {\ul 0.04195}    & 0.03957          & 0.03835       & \textbf{0.04665} & 11.20\%           \\
                                   & \multirow{2}{*}{NDCG \(\uparrow\)}   & 10 & 0.00630     & 0.00702       & {\ul 0.00822}      & 0.00807            & 0.00648             & 0.00709          & 0.00613          & 0.00599       & \textbf{0.00887} & 7.91\%            \\
                                   &                         & 30 & 0.00740     & 0.00842       & 0.01337            & {\ul 0.01355}      & 0.01128             & 0.01332          & 0.00928          & 0.00967       & \textbf{0.01431} & 5.61\%            \\ \bottomrule
\end{tabular}
}
\caption{The overall performance comparison results of applying our model and baselines on six real-world normal datasets. We evaluated the recommendation performance as a ranking task, underlined the best baseline result in each line, and put the best result in each line in bold; Higher Recall and NDCG mean better model performance. The arrow `\(\uparrow\)' (or `\(\downarrow\)') denotes that the higher (or lower) value means better performance on the metric. The `Imp.' row reports the relative improvement or decline of CSA-VAE against the best baseline. The result is calculated based on the mean of five repetitions with different random seeds for all models on each metric.
}
\label{tab:4}
\end{table}
\subsection{Comparision with Baselines (RQ2)}
The comparison between CSA-VAE and various baselines is shown in Table \ref{tab:3} and Table \ref{tab:4}. The best results (compared across two classes) are shown in bold, and the runner-ups are {\ul underlined}. In summary, we have the following observations:\\
\textbf{(1)} \textbf{\textcolor{black}{Consistent superior recommendation performance.}} The result demonstrates that the CSA-VAE model consistently outperforms the baselines regarding Recall and NDCG across various datasets and evaluation metrics. Specifically, CSA-VAE achieves the highest Recall and NDCG scores in nearly all cases, indicating its superior ability to recommend relevant items to users. Remarkably, CSA-VAE substantially improves Recall and NDCG compared to the baselines. The average performance boost against several state-of-the-art baselines achieves up to \textcolor{black}{9.55\%} across different datasets and evaluation settings, demonstrating the superiority of our model. \\
\textbf{(2)} \textbf{CSA-VAE can better model user preference.} As Table \ref{tab:3} shows, \textcolor{black}{on unbiased datasets \dataset{Coat}, \dataset{Yahoo R3}, \dataset{Reasoner} and \dataset{Kuairec}, CSA-VAE user only user-specific preference for the recommendation,} CSA-VAE achieves the highest Recall and NDCG scores in all cases, indicating its superior ability on model user preference rather than mixed preference in feedback data. Remarkably, CSA-VAE substantially improves Recall and NDCG compared to the deconfounded methods ICDF and Invpref, the average performance boost against several state-of-the-art baselines achieves up to \textcolor{black}{11.50\%}.\\
\textbf{(3)} \textbf{CSA-VAE can better fit the mixed preference in feedback data.} As Table \ref{tab:4} shows, on normal datasets \dataset{ML}-\dataset{100K}, \dataset{ML}-\dataset{1M}, \dataset{ML}-\dataset{10M}, \dataset{ML}-\dataset{20M} and \dataset{ML}-\dataset{25M}, and \dataset{Epinions}, CSA-VAE achieves the highest Recall and NDCG scores in all cases, indicating its superior ability on model user preference rather than mixed preference in feedback data. The slightly weaker performance on ML-1M compared to Rec-VAE is due to the fact that Rec-VAE uses additional training tricks, such as alternating updates of prior, which we did not use. The datasets with higher sparsity levels, such as ML-20M, Ml-25M, and Epinions, often present more challenges for traditional recommender systems due to the scarcity of user-item interactions. However, the CSA-VAE model demonstrates substantial improvements in these datasets, suggesting its robustness in handling sparse data and providing meaningful recommendations despite the challenges posed by data sparsity. \\
In summary, the comprehensive analysis of the performance of CSA-VAE across these datasets underscores its potential to significantly enhance recommendation quality and user engagement across a spectrum of real-world applications. The consistent outperformance over baseline models highlights the efficacy of integrating causal graph-based approaches in recommender systems, addressing unobserved confounders and providing more accurate and satisfying recommendations.
\begin{figure}[!t]
    \centering
    \includegraphics[width=0.9\linewidth]{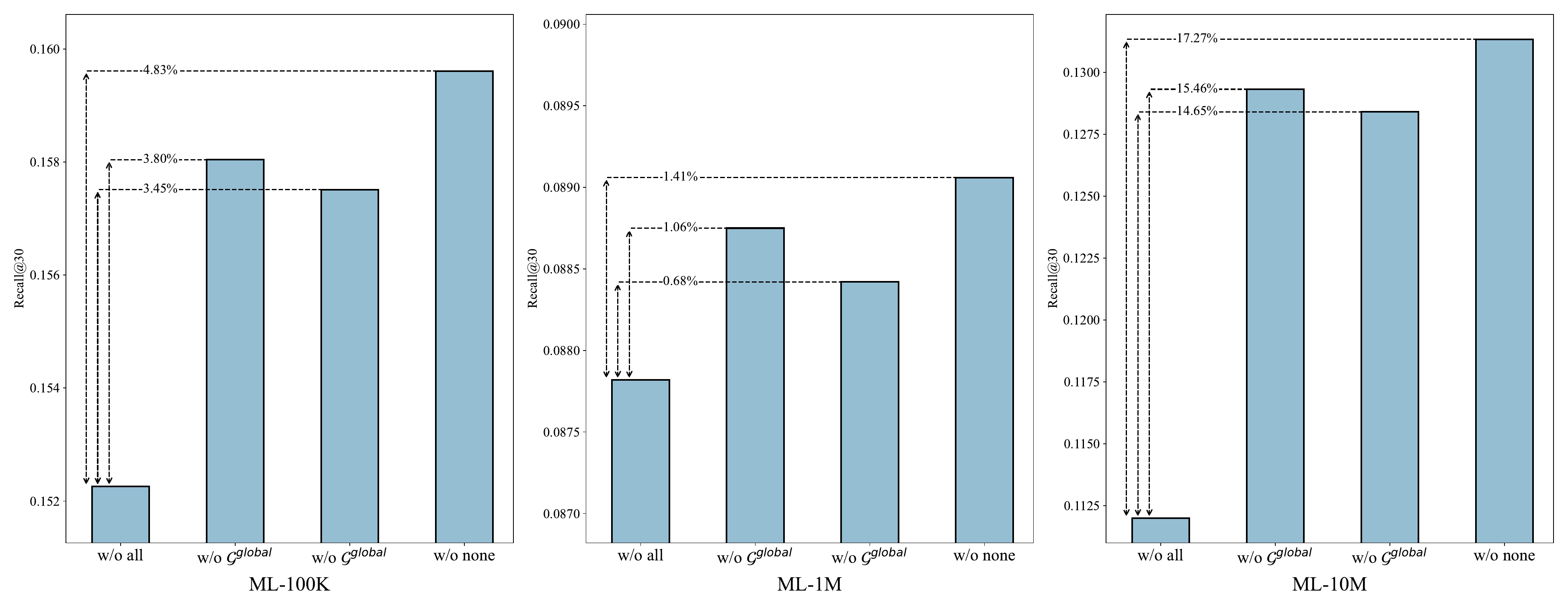}
    \caption{Ablation experiments on global and local graphs conducted on Ml-100k, ML-1M and ML-10M with Recall@30.}
    \label{fig:8}
\end{figure}
\subsection{Ablation Study (RQ3).}
\subsubsection{\textbf{Effectiveness of Global and Local Graphs.}} The Figure \ref{fig:8} presents results for different variants involving both the global and local graphs:$w/o\ \mathcal{G}^{none}$ without both global and local graphs, equivalent to Multi-VAE; $w/o\ \mathcal{G}^{local}$ (without local graph); $w/o\ \mathcal{G}^{global}$ (without global graph); and $w/o\ \mathcal{G}^{none}$ (with both global and local graphs). we have the following observations: \\
\textbf{(1)} Using either the global graph or the local graph alone results in improved performance on the dataset. The global graph captures the macro-level causal relationships among confounders, albeit losing specificity to users. On the other hand, the local graph captures user specificity but loses the accurate relationships between confounders. Both contribute partially to the causal relationships between confounders, leading to an enhancement in model performance. \\
\textbf{(2)} Simultaneously using the local and global graphs results in a performance improvement greater than the sum of their individual contributions. The global graph can correct the erroneous causal relationships between confounders in the local graph, while the local graph assigns user-specific weights to the global graph. The synergy of both significantly enhances the model's performance. In summary, the collaborative effect of both graphs significantly enhances recommendation quality by enabling a more comprehensive understanding of user behavior and preferences, creating a powerful model for accurate and effective recommendations.

\begin{figure}[tp]
   \centering
    \includegraphics[width=0.9\linewidth]{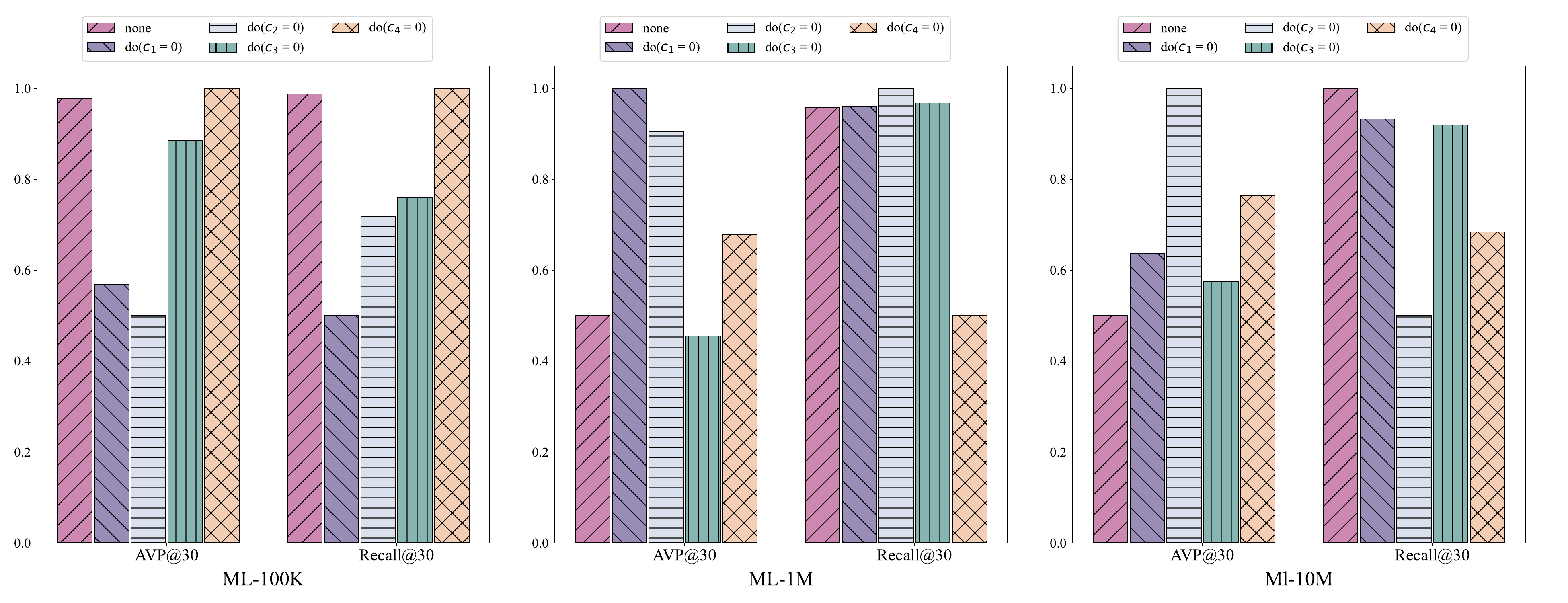}
    \caption{"do-operation" experiments conducted on Ml-100k, ML-1M and ML-10M with Recall@30 and AVP@30, the number of confounders, $k = 4$. \(do(c_i = 0)\) means performing the "do-operation" on $c_1$, $c_2$, $c_3$, and $c_4$, and \(do(none)\) means no "do-operation" on any confounders.}
    \label{fig:10}
\end{figure}
\subsubsection{\textbf{"do-operation" with Mask Graph.}} We conducted common "do-operation" experiments in the causal inference domain on the \dataset{ML}-\dataset{100K}, \dataset{ML}-\dataset{1M}, and \dataset{ML}-\dataset{10M} datasets, and the results are shown in Figure \ref{fig:10}. We employed four concepts of confounders ($k=4$) to model the unobserved confounders in these datasets. We used the popularity of items recommended (AVP@30) to the users and Recall@30 to evaluate the effect of the confounders. It is important to note that applying the "do-operation" to a confounder $c_i$ involves using an additional masking matrix and performing a dot product with the global graph. In the masking matrix, the $i$-th row is all zeros (indicating no influence as a parent node), and the rest of the rows are all ones. With this operation, we can answer the question: "If confounder \(c_k\) has no effect, how does this impact user performance?" From the Figure \ref{fig:10}, we have the following observation:\\
\textbf{(1)} \textbf{Confounder do harmful to user preference modeling.} When we operate on confounders, the performance of the model is mostly affected, which means that the confounders play an important role in the final recommendation process, where user-specific preferences should be absolutely dominant. The existence of confounders leads to the inevitable influence of confounders in the process of modeling user-specific preferences. Therefore, decoupling confounders and user-specific preferences is a necessary approach to better model user preferences.\\
\textbf{(2)} \textbf{Confounder not totally harmful to user preference modeling.} When we operate on confounders on the \dataset{ML}-\dataset{1M} and \dataset{ML}-\dataset{10M} datasets, as Figure \ref{fig:10} shows, the popularity of recommend items increases, which means these confounders have a positive effect on mitigating popularity bias. Based on this finding, we can achieve the short-term goals of the recommendation system by controlling variables. For instance, we can recommend popular items to users at certain times while refraining from using additional recommendation strategies at other times.\\

\begin{figure}
    \centering
    \includegraphics[width=0.9\linewidth]{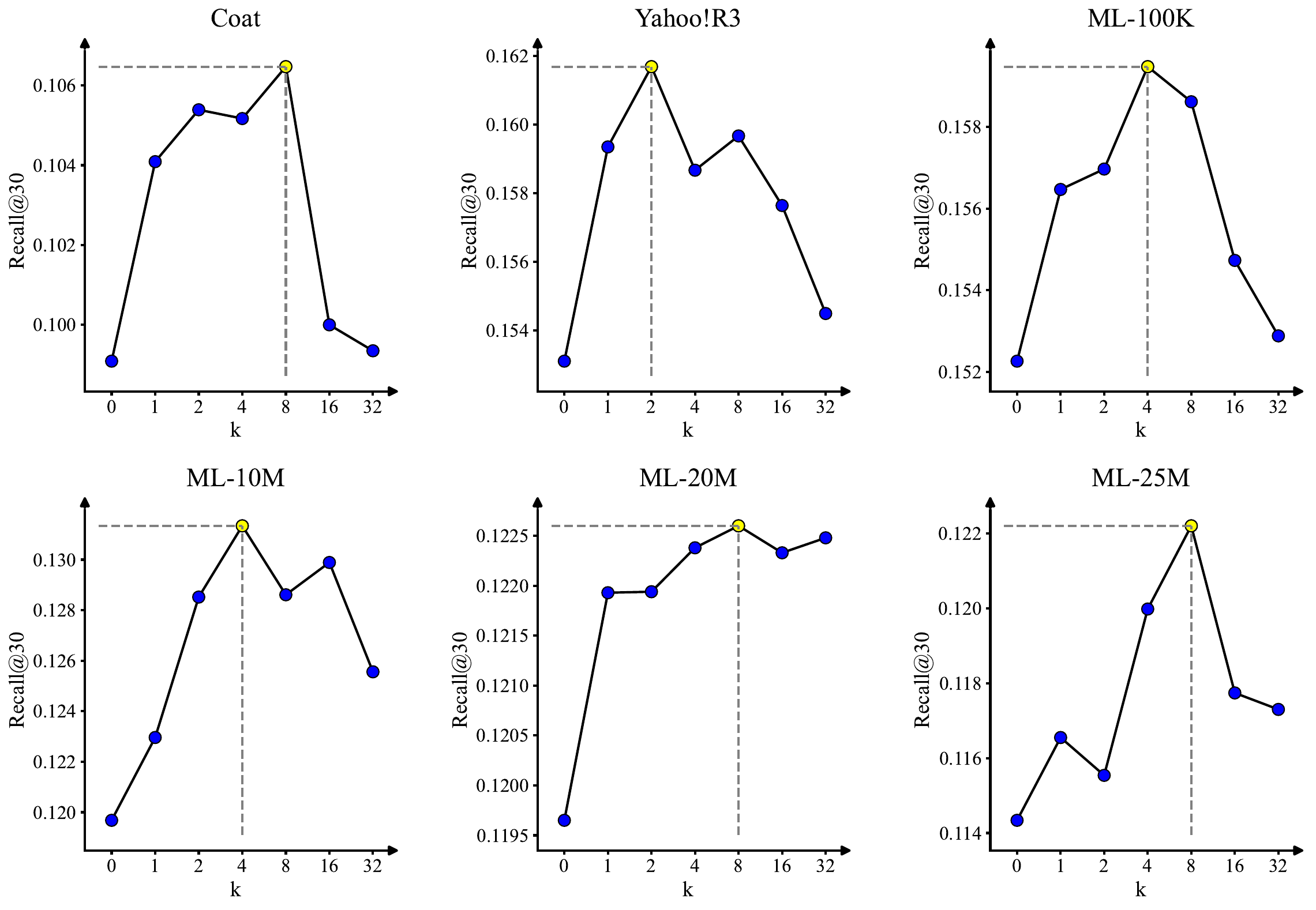}
    \caption{Sensitivity of CSA-VAE with different confounders number \(k\). The horizontal axes of all sub-figures are the variable \(k\).}
    \label{fig:7}
\end{figure}
\begin{table}[!t]
\centering
\renewcommand{\arraystretch}{1.5}
\resizebox{0.8\linewidth}{!}{
\begin{tabular}{@{}lcccccc@{}}
\toprule
\textbf{Model}      & \textbf{CDAE} & \textbf{Multi-DAE} & \textbf{Multi-VAE} & \textbf{Macrid-VAE} & \textbf{Rec-VAE} & \textbf{CSA-VAE}   \\ \midrule
\textbf{Complexity} & \textit{O(n(m+d))}     & \textit{O(n(m+d))}          & \textit{O(n(m+d))}          & \textit{O(k\(\cdot\)n(m+d))}         & \textit{O(n(m+d))}        & \textit{O(n(m+\(k^{2}\)+d))} \\
\textbf{FLOPs (M)} & 23.11         & 29.32              & 29.33              & 51.03               & 44.02            & 29.64              \\ \bottomrule
\end{tabular}
}
\caption{Computational Complexity Comparison on ML-25M.}
\label{tab:compleana}
\end{table}
\begin{figure}[!t]
    \centering
    \begin{minipage}[h]{\linewidth}
        \centerline{\includegraphics[width=0.6\linewidth]{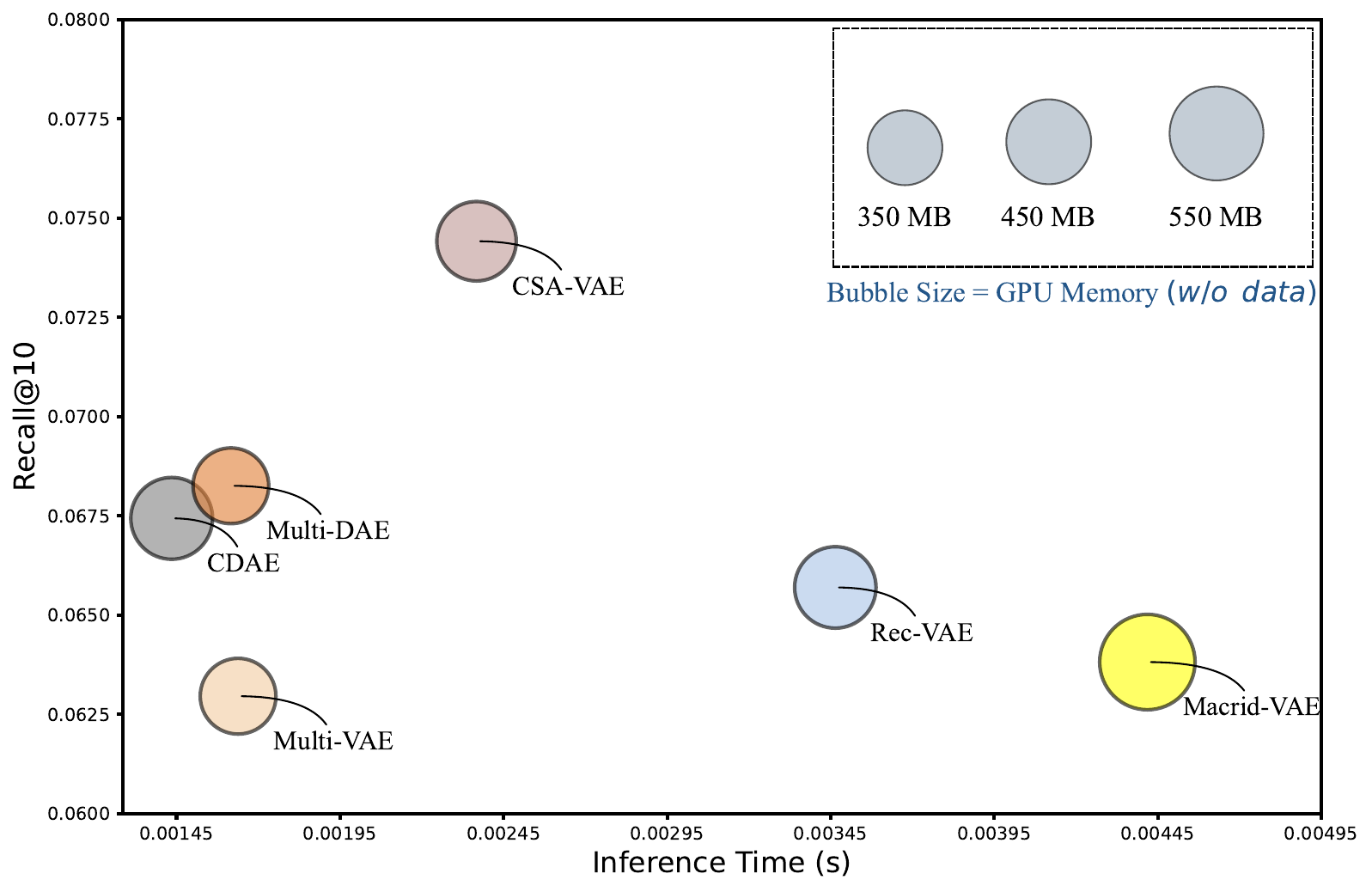}}
        \centerline{(a) Recall / Efficiency traded-off on ML-25M.}
    \end{minipage}
    \vfill
    \vspace{0.8cm}
    \begin{minipage}[!h]{0.48\linewidth}
        \centerline{\includegraphics[width=0.8\linewidth]{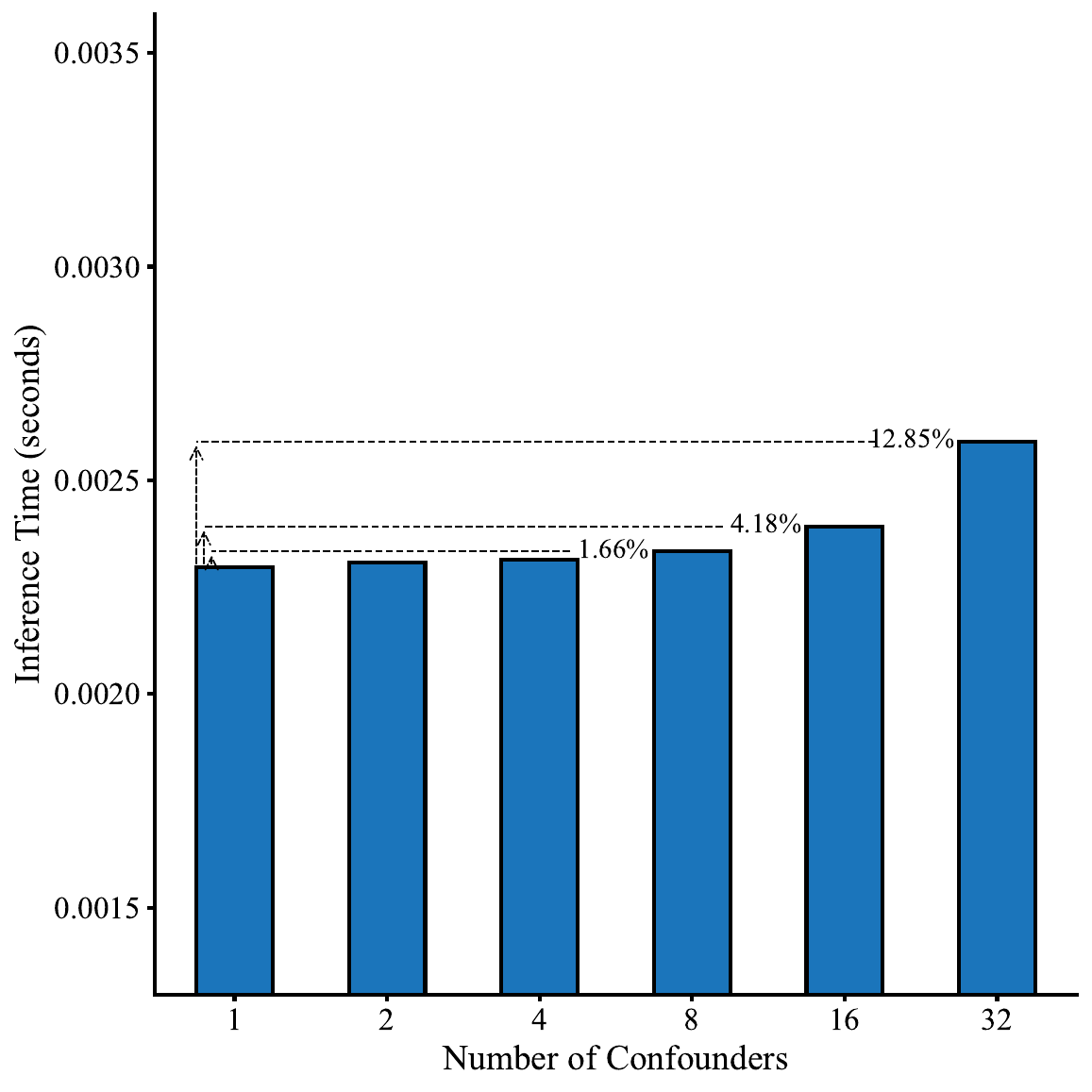}}
        \centerline{(b) CSA-VAE Inference Time}
    \end{minipage}
    \hfill
    \begin{minipage}[!h]{0.48\linewidth}
        \centerline{\includegraphics[width=0.8\linewidth]{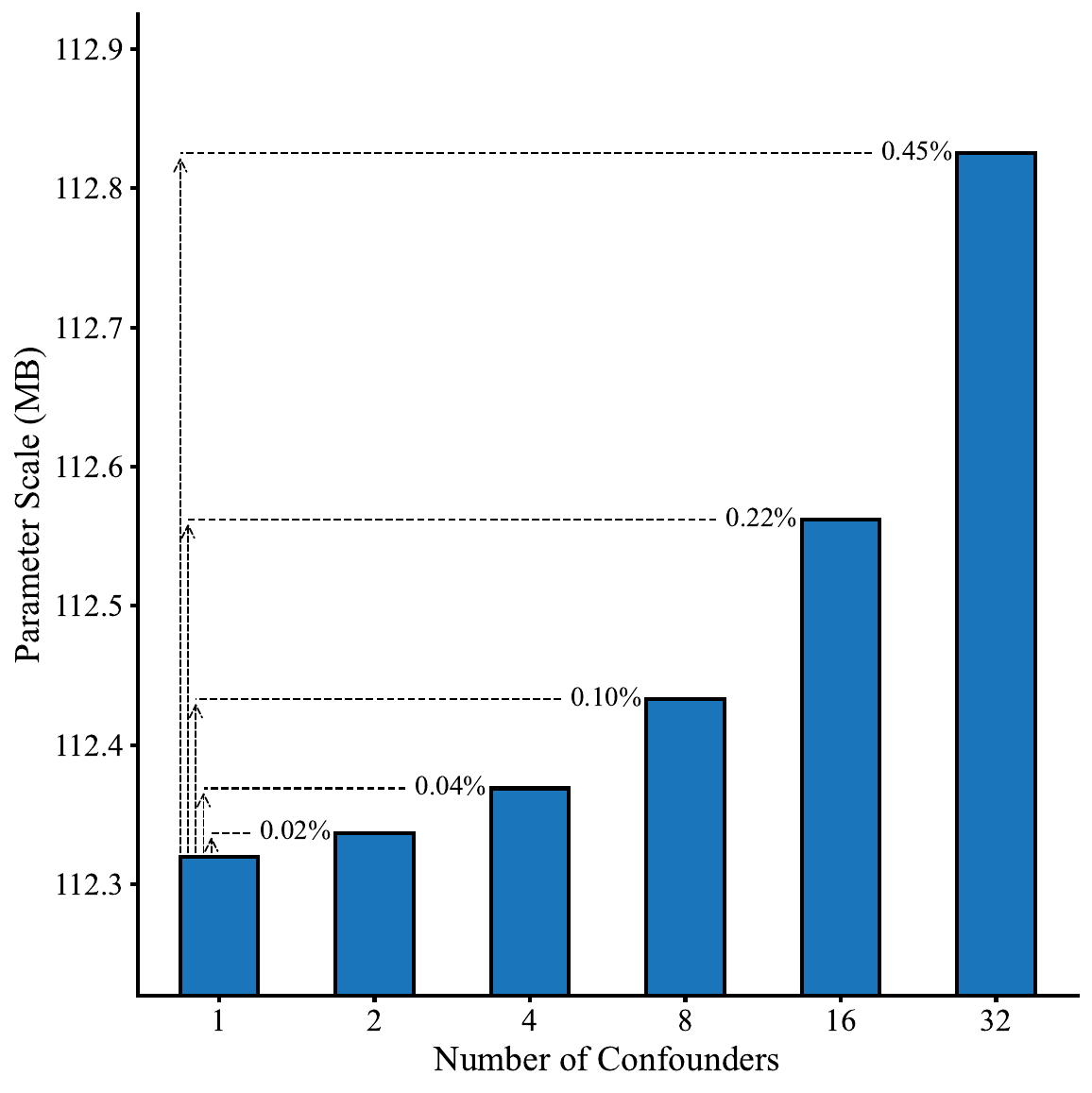}}
        \centerline{(c) CSA-VAE Parameter Scale}
    \end{minipage}
    \caption{Computational Complexity Comparison on ML-25M. (a) is the Recall / Efficiency traded-off comparison. A higher position on the vertical axis indicates better performance, while moving left along the horizontal axis signifies lower inference time costs. A smaller bubble size indicates reduced GPU memory costs during inference. (b) and (c) compare memory and inference time under various value numbers of confounders.}
    \label{fig:complexana}
\end{figure}

\subsection{Computational Complexity Analysis on ML-25M}
\textcolor{black}{
We computed the model's Floating Point Operations (FLOPs) on the biggest dataset \dataset{ML}-\dataset{25M} using MMEengine \cite{mmengine2022} to indicate the model's learning complexity, which is a critical component in optimizing neural networks for performance and efficiency. The higher number of FLOPs signifies a more challenging learning process. The results are presented in Table \ref{tab:compleana} and Figure \ref{fig:complexana} (a). Additionally, we conducted experiments with varying numbers of confounders to validate their effect on the model's inference cost on \dataset{ML}-\dataset{25M}. The corresponding results are shown in Figures \ref{fig:complexana} (b) and (c). Based on these experiments, we have the following observations:
\begin{itemize}
    \item Compared to classical models, our model has a relatively lower increase in learning difficulty and outperforms other mainstream VAE-based models. Results shown in Table \ref{tab:compleana}, introducing the causal graph as an additional learning task increases the learning difficulty of CSA-VAE. However, compared to Rec-VAE and Macrid-VAE, the increase in FLOPs is relatively lower.
    \item Our model balances model size and inference speed to a certain extent. As shown in Figure \ref{fig:complexana} (a), compared to mainstream VAE-based models (e.g., Rec-VAE), our model achieves higher inference speed and has significantly smaller model parameters. Notably, compared to traditional models (e.g., CDAE), our model also shows a clear advantage in terms of parameter size.
    \item Under the condition of maintaining recommendation performance, the increase in model inference time caused by adding confounders remains within an acceptable range. As shown in Table \ref{tab:compleana}, the number of confounders is a key determinant of the higher time complexity of CSA-VAE compared to other models, such as Multi-VAE. However, when the dataset is large, \(m\) becomes much more extensive than \(k\). At this point, the model's inference time is primarily influenced by the dataset size, with the number of confounders having a negligible impact. This observation is confirmed in Figure \ref{fig:complexana} (b). When \(k = 8\), the inference time is only 1.66\% higher than when \(k = 1\), which falls within an acceptable range. It is also important to note, as shown in Figure 9, that tremendous values of \(k\) do not yield better results. Therefore, a value of \(k \leq 8\) is generally considered optimal.
    \item In our model, the number of confounders minimally impacts the model's parameter size. As shown in Figure \ref{fig:complexana} (c), when \( k = 32 \), the model size increases by only 0.45\% compared to when \( k = 1 \), this increase in model size is almost imperceptible relative to the overall size of the model.
\end{itemize}
}
\begin{figure}[!t]
    \centering
    \includegraphics[width=0.8\linewidth]{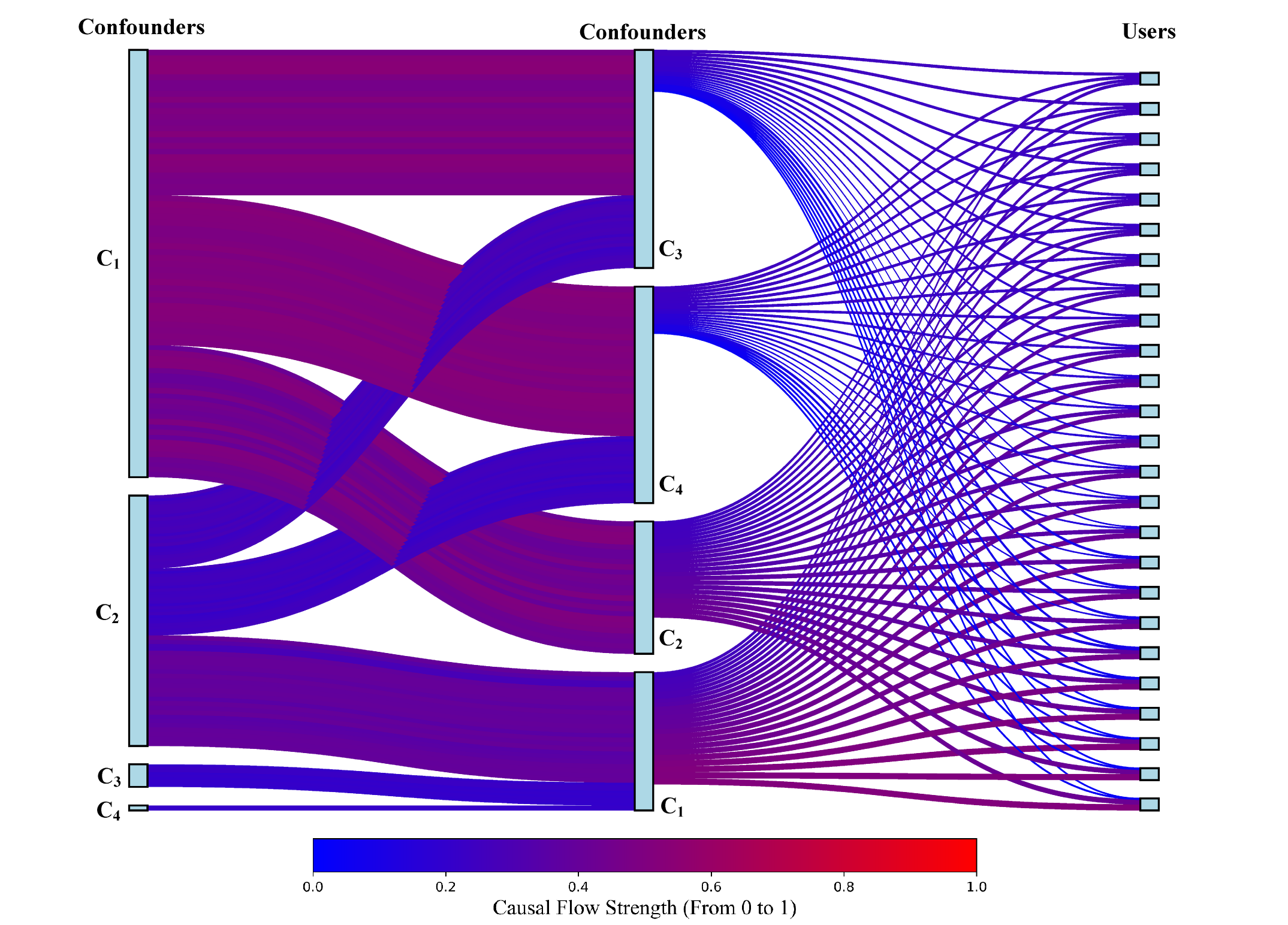}
    \caption{Causal flow visualization based on Sankey (left to right). The left and middle columns represent the causal flow between confounders, and the right columns represent the causal flow between confounders and users (30 randomly selected). The color intensity indicates the strength of the causal relationship, and closer to red indicates more muscular strength (Since there is no intensity clipping, form weak loops (lower strength value) in the visualization).}
    \label{fig:causalflow}
\end{figure}
\subsection{Causal Flow Visualization}
\textcolor{black}{
We use Sankey to represent the flow and strength of causal relationships between confounders and between confounders and users. The results are shown in Figure \ref{fig:causalflow}, we have the following observations:}\par
\textcolor{black}{
\begin{itemize}
    \item The causal strengths between the confounders in the two columns on the left are not identical (the strength between the same two nodes varies from user to user), which should be the same if only using the global causal graph (without user information guidance), confirming that our model can generate user-specific causal graphs and capture complex nonlinear relationships between confounders under the guidance of user preference information.
    \item The causal strength exhibits clear user differences when it flows to users through confounders. The same confounder has varying levels of influence on different users; results demonstrate that our model captures the complex nonlinear relationships between confounders and user preferences.
    \item Users are affected to varying levels by different confounders, and the same user is affected to different capacities by multiple confounders, indicating that our model captures the varying sensitivity of users to different confounders.
\end{itemize}
}
\subsection{Sensitivity Analysis (RQ4).}
We used various values of $k$ on the Coat, Yahoo!R3, ML-100k, ML-10M, ML-20M, and ML-25M datasets to verify the influence of the number of confounders. As Figure \ref{fig:7} (right side) shows, we have the following observations: \\
\textbf{(1)} \textbf{The performance of different numbers of confounders is closely related to the magnitude of the datasets.} As Figure \ref{fig:7} (right side) shows, With the increase in the number of confounders, the model performance decreases more on small data sets than on large data sets. This result arises because the dataset size is insufficient to fully support the model in identifying edges in the directed causal graph corresponding to the number of confounders. Specifically, for a directed causal graph with $k$ nodes, there are \(2^{k(k-1)}\) possible edges, requiring at least \(2^{k(k-1)}\) interactions to identify these edges reliably, and we need more due to the user-specific in the recommendation. Consequently, as the number of nodes increases significantly, inadequate data prevents the model from producing an improved graph, leading to a decline in performance.\\
\textbf{(2)}  \textbf{The greater the number of confounders the better.}  A larger $k$ implies a more diverse set of confounders. As $k$ increases, we obtain a finer-grained representation of the confounders, improving model performance. As Figure \ref{fig:7} shows, the model's performance increases with the increase of \(k\) until reaching the critical point mentioned above. However, an excessively large $k$ may cause the learned representation of confounders to be more than the actual number of concepts influencing the data, leading to model overfitting and a subsequent decrease in performance.
\begin{figure}[!t]
    \centering
    \begin{minipage}[!h]{0.48\linewidth}
        \centerline{\includegraphics[width=0.75\linewidth]{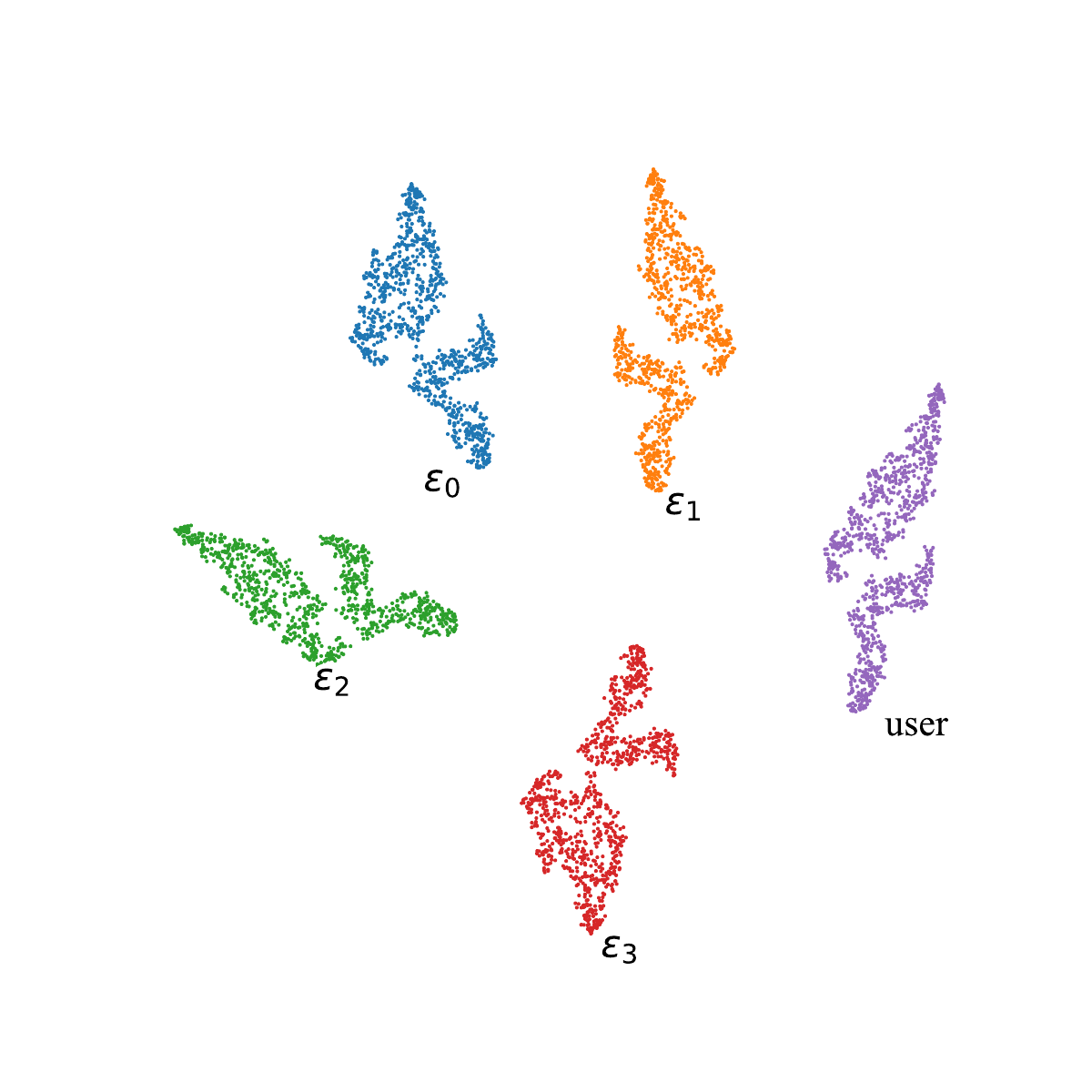}}
        \centerline{(a) Before \textsc{Causal Layer}}
        \centerline{\quad \quad with the diversity constraint}
    \end{minipage}
    \hfill
    \begin{minipage}[!h]{0.48\linewidth}
        \centerline{\includegraphics[width=0.75\linewidth]{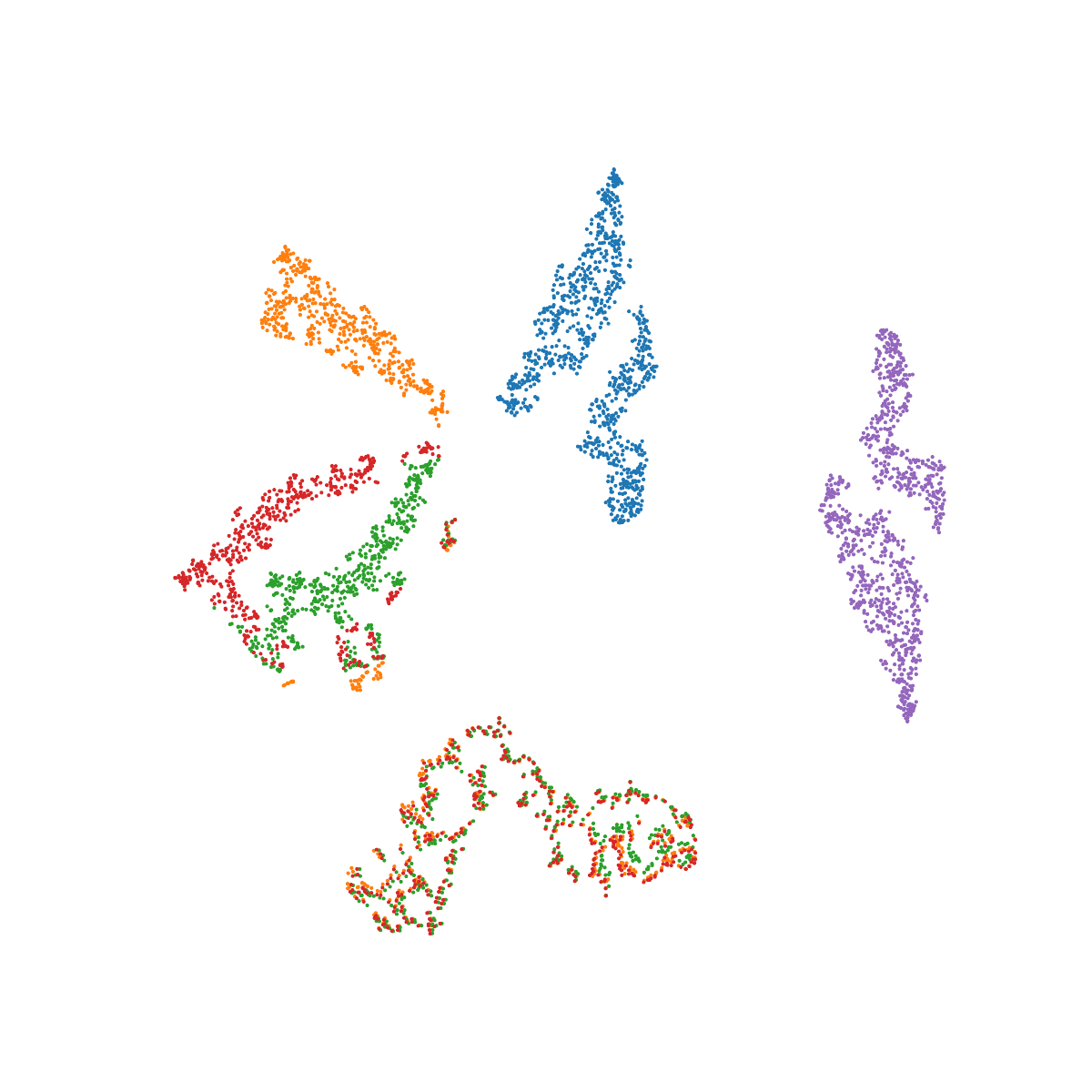}}
        \centerline{(b) After \textsc{Mask Layer}}
        \centerline{\quad \quad with the diversity constraint}
    \end{minipage}
    \vfill
    \vfill
    \begin{minipage}[!h]{0.48\linewidth}
        \centerline{\includegraphics[width=0.75\linewidth]{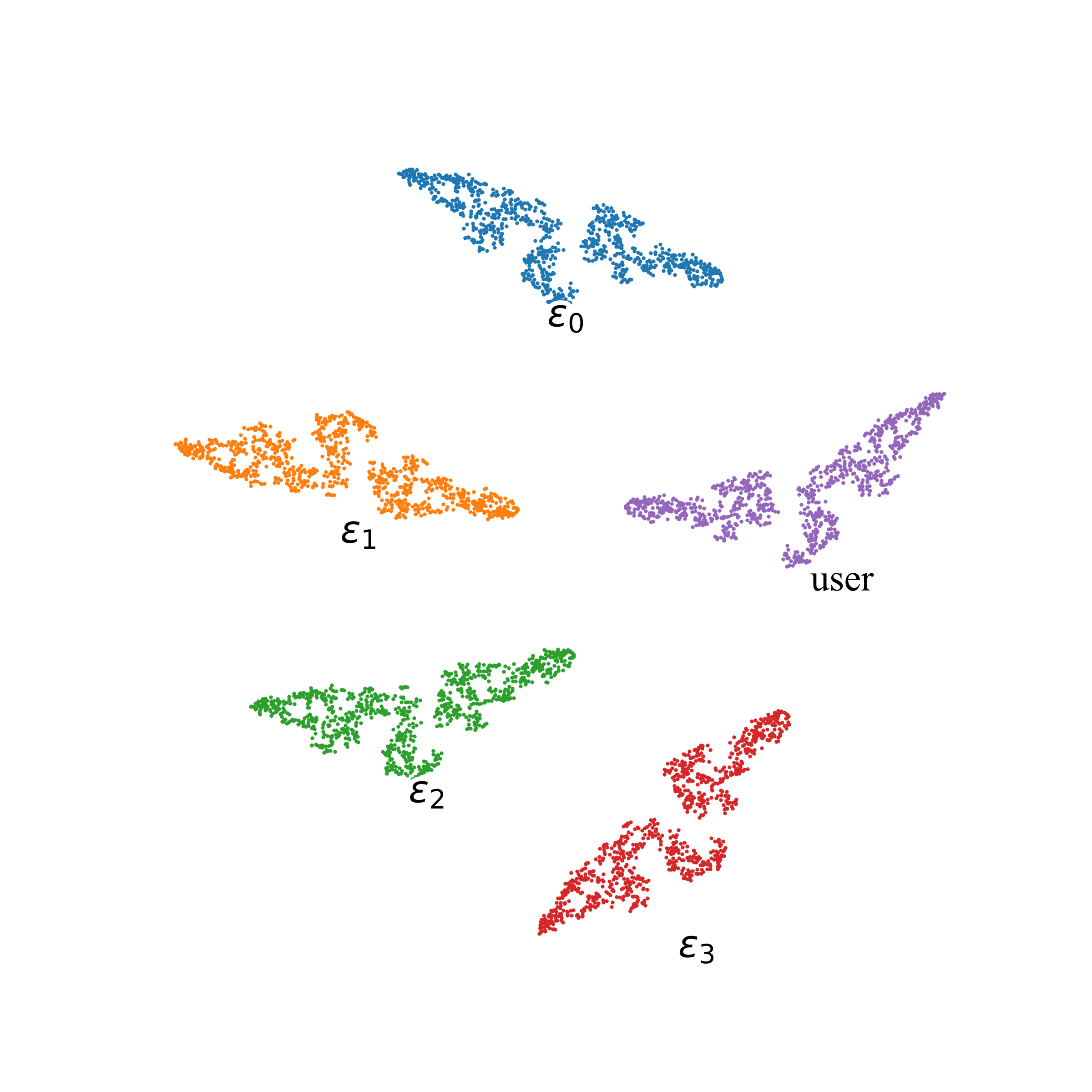}}
        \centerline{(c) Before \textsc{Causal Layer}}
        \centerline{\quad \quad without the diversity constraint}
    \end{minipage}
    \hfill
    \begin{minipage}[!h]{0.48\linewidth}
        \centerline{\includegraphics[width=0.75\linewidth]{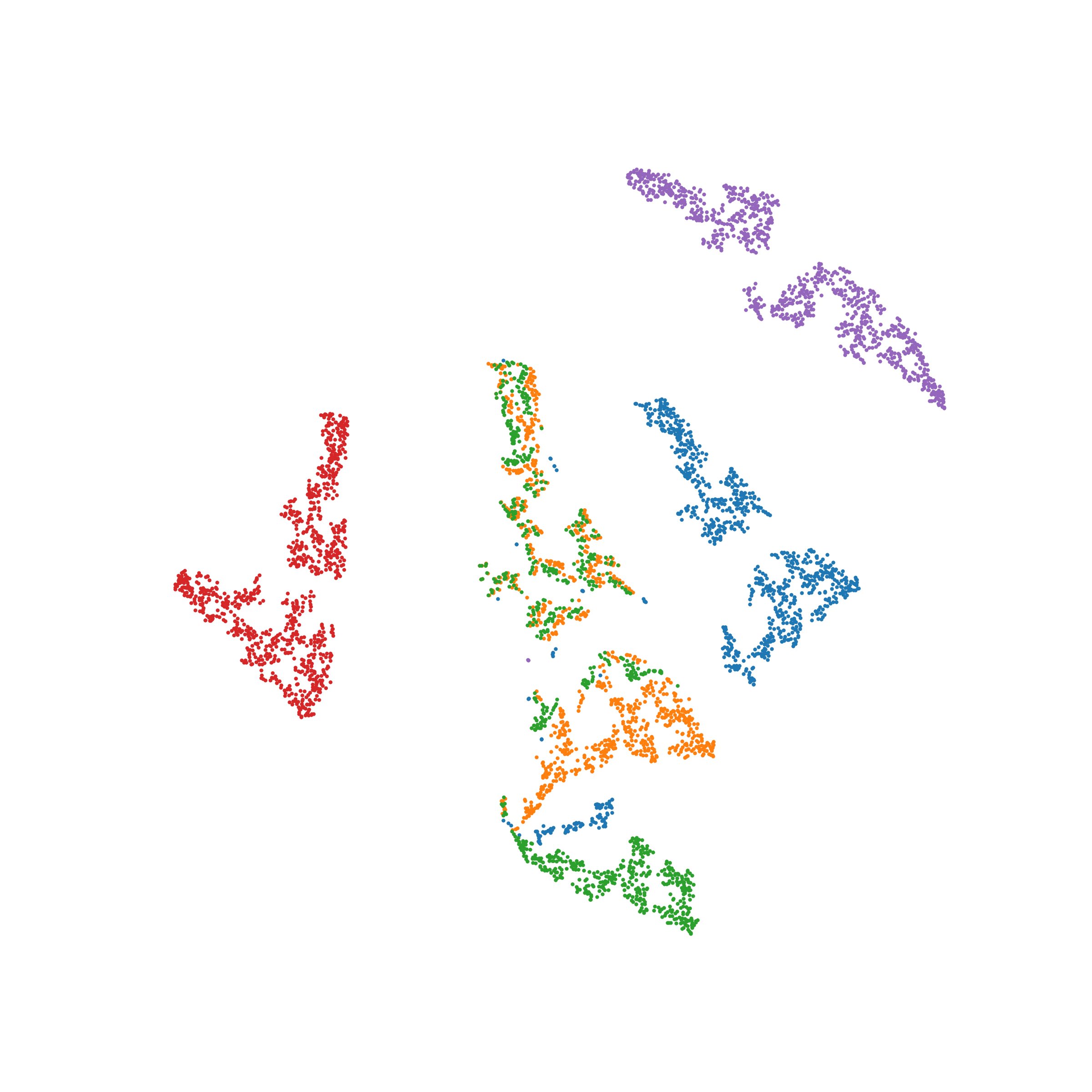}}
        \centerline{(d) After \textsc{Mask Layer}}
        \centerline{\quad \quad without the diversity constraint}
    \end{minipage}
    \caption{Visualization of confounders and user preference on ML-100K, $k = 4$. \textcolor{black}{The same colors in (a)-(b) and (c)-(d) represent the same confounders. (a)-(c)is the visualization of exogenous variables \(\epsilon_k\) and user preference before \(\mathsf{Causal\ Layer}\) and (b)-(d) is the visualization of confounders and user preference after \(\mathsf{Mask\ Layer}\)}. The numbers \(\{0,1,2,3\}\) correspond to the elements in the confounders set \(\{c_1,c_2,c_3,c_4\}\), and \(\epsilon_i\ (i \in [0,k])\) identifies the exogenous variables of the confounders. \textcolor{black}{Among them, (a)–(b) are the results of training under diversity constraints, while (c)–(d) are the results of training without diversity constraints.}}
    \label{fig:9}
\end{figure}

\subsection{Visualization of Confounders and User-specific preference}
We use t-SNE \cite{tsne} to visualize the confounders and user preference on ML-100K, with $k = 4$. Specifically, we first visual the exogenous variables of confounders and user preference before \(\mathsf{Causal\ Layer}\) and then visual the representation of confounders and user preference after \(\mathsf{Mask\ Layer}\). We have the following observations:
\textcolor{black}{
\begin{itemize}
   \item (1) As shown in Figure \ref{fig:9}, the exogenous variables \(\epsilon_k\) of \(k\) confounders and user preferences exhibit five distinct clusters in both (a) and (c), confirming the independence between confounders and user preferences, thus supporting our Assumption 1.
    \item (2) After encoding global and local graphs, the causal relationship involving confounders manifests as adjacent clusters in the graph in both (b) and (d). Results demonstrate that our method CSC-VAE effectively captures causal relationships between confounders, thereby enhancing user-specific and mixed preferences modeling in feedback data.
    \item (3) Regardless of the presence of diversity constraints, the exogenous variables of the confounders in (a) and (c) form independent clusters, demonstrating that the model can capture the independence of the exogenous variables of the confounders, which is not a result of the diversity constraints. Additionally, from (b) and (d), it can be observed that the confounders after the mask layer integrated better when diversity constraints are applied, indicating that the presence of diversity constraints helps the model obtain a better representation of the confounders and uncover the causal relationships between them.
\end{itemize}
}
\begin{figure}[!t]
    \centering
    \includegraphics[width=0.9\linewidth]{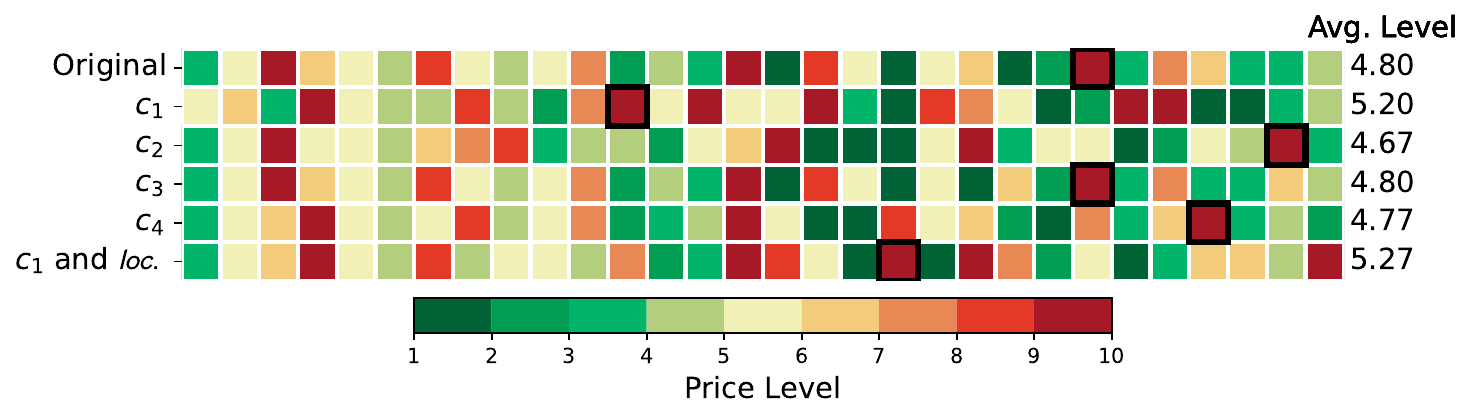}
    \caption{The case study indicates the effectiveness of controlling confounders in recommender systems. We divided the prices in the \dataset{Epinions} dataset into ten equally sized intervals, with higher levels corresponding to higher prices. The selected user prefers high-priced items in his feedback data, and the black box represents the ground-truth item.}
    \label{fig:casestu}
\end{figure}
\subsection{Case Study on Epinions}
\textcolor{black}{Using the \dataset{Epinions} dataset, which includes the price feature, we selected a user with a significant number of high-priced items in his historical interactions as the subject of the case study. By controlling for confounders, we generated alternative lists of recommended items for this user. We have the following observations:
\begin{itemize}
    \item We can obtain a more accurate user recommendation list by controlling the confounders. As shown in Figure \ref{fig:casestu}, controlling the confounder \(c_1\) improves the average price level of the recommendation list and the ranking of the ground-truth item.
    \item Only the global causal graph is relied upon without using the local causal graph, achieving a suboptimal recommendation list. As shown in Figure \ref{fig:casestu}, masking the local causal graph still improves the overall price level of the recommendation list, but the ranking of the ground-truth item remains suboptimal.
    \item When users are dissatisfied with the system's current recommendations, our model allows them to modify their recommendation list by controlling the confounders and causal graphs learned by the model. For example, when a user believes that the recommendation system's modeling expectations do not align with their preferences (e.g., a user sensitive to price), they can adjust the confounders to increase or decrease the average price level of the recommendations. Results demonstrate that our model, compared to traditional models, offers users greater flexibility, potentially improving overall user satisfaction.
\end{itemize}}

\section{Conclusion}
Predicting user preferences in the presence of confounders is a challenging problem. We first redefined the problem, incorporating the influence of confounders into the model. We proposed a mild assumption to separate user preferences from confounders and used a combination of local and global graphs to capture the causal relationships between confounders and user-specific preferences. Finally, we proposed a VAE-based model called CSA-VAE. Extensive experiments are conducted on a synthetic dataset and nine real-world datasets to demonstrate the model's superiority. We theoretically proved the model's Evidence Lower Bound (ELBO) and the learned graph's identifiability. Furthermore, we employed the "do-operation" method to validate the controllability of the model, potentially offering users fine-grained control over the objectives of their recommendation lists with the learned causal graphs. Future work can explore advanced unsupervised clustering methods to obtain categories of confounders further, addressing the limitation of uncertainty in the impact categories of obtained confounders on user preferences.
\section{ACKNOWLEDGMENTS}
 This work is supported by the Natural Science Foundation of China No. 62472196, Jilin Science and Technology Research Project 20230101067JC, and the National Natural Science Foundation of China under Grant NO. 62176014, the Fundamental Research Funds for the Central Universities.
\bibliographystyle{ACM-Reference-Format}
\bibliography{sample-base}

\appendix

\section{Proofs}
\subsection{Proof of Evidence Lower Bound}
\label{elbo}
\begin{equation}
   \begin{aligned}
        \ln p_{\theta}(x_{u}) \geq E_{q(z,|x_u,\mathit{sInfo_u},\textbf{C})}\left[\ln p(x_{u} | z, \mathit{sInfo_u}, \textbf{C})]\right. - D_{KL}(q(z|x_u,\mathit{sInfo_u}, \textbf{C}) \| p(z|\mathit{sInfo_u},\textbf{C})). \nonumber\\
    \end{aligned}
\end{equation}
We give the proofs as follows:
\begin{proof}
The mixed preference \( z \) is obtained from the conditional distribution \( p(z|\mathit{sInfo_u}, \textbf{C}) \), and \( \mathit{sInfo_u} \) and \( \textbf{C} \) are independent of each other, we need to adjust the derivation of the VAE's ELBO to account for this new conditional distribution.

We can write the conditional distribution \( p(z|\mathit{sInfo_u}, \textbf{C}) \) and incorporate these conditions into the ELBO derivation.

The marginal log-likelihood of the data can be expressed as:

\[
\log p(x_u|\mathit{sInfo_u}, \textbf{C}) = \log \int p(x_u, z|\mathit{sInfo_u}, \textbf{C}) \, dz.
\]

We introduce the variational distribution \( q(z|x_u, \mathit{sInfo_u}, \textbf{C}) \) to approximate the posterior distribution \( p(z|x_u, \mathit{sInfo_u}, \textbf{C}) \), and use Jensen's inequality to derive the ELBO:

\[
\log p(x_u|\mathit{sInfo_u}, \textbf{C}) = \log \int q(z|x_u, \mathit{sInfo_u}, \textbf{C}) \frac{p(x_u, z|\mathit{sInfo_u}, \textbf{C})}{q(z|x_u, \mathit{sInfo_u}, \textbf{C})} \, dz.
\]

\[
\log p(x_u|\mathit{sInfo_u}, \textbf{C}) = \log \mathbb{E}_{q(z|x_u, \mathit{sInfo_u}, \textbf{C})} \left[ \frac{p(x_u, z|\mathit{sInfo_u}, \textbf{C})}{q(z|x_u, \mathit{sInfo_u}, \textbf{C})} \right] \geq \mathbb{E}_{q(z|x_u, \mathit{sInfo_u}, \textbf{C})} \left[ \log \frac{p(x_u, z|\mathit{sInfo_u}, \textbf{C})}{q(z|x_u, \mathit{sInfo_u}, \textbf{C})} \right].
\]

The expectation on the right is the ELBO:

\[
\mathbb{E}_{q(z|x_u, \mathit{sInfo_u}, \textbf{C})} \left[ \log \frac{p(x_u, z|\mathit{sInfo_u}, \textbf{C})}{q(z|x_u, \mathit{sInfo_u}, \textbf{C})} \right] = \mathbb{E}_{q(z|x_u, \mathit{sInfo_u}, \textbf{C})} \left[ \log p(x_u, z|\mathit{sInfo_u}, \textbf{C}) - \log q(z|x_u, \mathit{sInfo_u}, \textbf{C}) \right].
\]

We decompose the joint distribution \( p(x_u, z|\mathit{sInfo_u}, \textbf{C}) \) as follows:

\[
p(x_u, z|\mathit{sInfo_u}, \textbf{C}) = p(x_u|z, \mathit{sInfo_u}, \textbf{C}) p(z|\mathit{sInfo_u}, \textbf{C}).
\]

Thus, the ELBO can be further decomposed as:

\begin{equation}
    \begin{aligned}
        &\mathbb{E}_{q(z|x_u, \mathit{sInfo_u}, \textbf{C})} \left[ \log p(x_u, z|\mathit{sInfo_u}, \textbf{C}) - \log q(z|x_u, \mathit{sInfo_u}, \textbf{C}) \right]\\
        =& \mathbb{E}_{q(z|x_u, \mathit{sInfo_u}, \textbf{C})} \left[ \log p(x_u|z, \mathit{sInfo_u}, \textbf{C}) + \log p(z|\mathit{sInfo_u}, \textbf{C}) - \log q(z|x_u, \mathit{sInfo_u}, \textbf{C}) \right] \nonumber.
    \end{aligned}
\end{equation}

We can rewrite it as the sum of two parts:

\[
\mathbb{E}_{q(z|x_u, \mathit{sInfo_u}, \textbf{C})} \left[ \log p(x_u|z, \mathit{sInfo_u}, \textbf{C}) \right] - D_{KL}(q(z|x_u, \mathit{sInfo_u}, \textbf{C}) \| p(z|\mathit{sInfo_u}, \textbf{C})),
\]

where \(D_{KL}(q(z|x_u, \mathit{sInfo_u}, \textbf{C}) \| p(z|\mathit{sInfo_u}, \textbf{C}))\) is the Kullback-Leibler divergence between the variational distribution \( q(z|x_u, \mathit{sInfo_u}, \textbf{C}) \) and the conditional prior distribution \( p(z|\mathit{sInfo_u}, \textbf{C}) \).

In summary, when \( z \) is obtained from the conditional distribution \( p(z|\mathit{sInfo_u}, \textbf{C}) \), and \( \mathit{sInfo_u} \) and \( \textbf{C} \) are independent, the ELBO of the VAE is:

\[
\text{ELBO} = \mathbb{E}_{q(z|x_u, \mathit{sInfo_u}, \textbf{C})} \left[ \log p(x_u|z, \mathit{sInfo_u}, \textbf{C}) \right] - D_{KL}(q(z|x_u, \mathit{sInfo_u}, \textbf{C}) \| p(z|\mathit{sInfo_u}, \textbf{C})).
\]
\end{proof}
\subsection{Synthetic Dataset Details.}
\label{sydata}
For synthetic data experiments, the number of user samples is 300, and for each user sample, 500 items. We assume that users are influenced by four different categories of confounders, and the causal relationships of these four confounders can be generated using Eq \ref{eq:2}. Specifically, we consider the following causal structural model:
\begin{equation}
\centering
    \begin{aligned}
        n_1 &= \mathcal{N}(\lambda_{1},\beta_{1}),\\
        n_2 &= \mathcal{N}(\lambda_{2},\beta_{2}),\\
        n_3 &= \mathcal{N}(\lambda_{3},\beta_{3}),\\
        n_4 &= \mathcal{N}(\lambda_{4},\beta_{4}),\\
        c_1 &= n_1,\\
        c_2 &= w_2(u)(c_1) + n_2,\\
        c_3 &= w_3(u)(c_2) + n_3,\\
        c_4 &= w_4(u)(c_2) + n_4,\\
    \end{aligned}
\end{equation}
where $\lambda_{i}$ and $\beta_{i}$ are the mean and variance of the Gaussian distribution, we sample the $\lambda_{i}$ and $\beta_{i}$ from the uniform distribution $[-3,3]$ and $[0.01,4]$, respectively. We sample the user-specific weight $w_i$ from Poisson distribution with the give $u$:
\begin{equation}
    \begin{aligned}
        u = \mathcal{N}(0,1).
    \end{aligned}
\end{equation}
Finally, given the set of confounders and users, we can generate the final observed value $\mathcal{x}$ through a non-linear function $g(\cdot)$:
\begin{equation}
    \begin{aligned}
        \mathcal{X} = g(c_1,c_2,c_3,c_4,u).
    \end{aligned}
\end{equation}
In our experiments, we used a two-layer MLP for the blending generation.
\end{document}